%% file: main.tex
\documentclass[11pt]{eptcs}

\input{\macrospath/ben-macros-preamble}

\input{\macrospath/ben-macros-paper}

\input{\macrospath/ben-macros-diagrams}

\input{local-macros}

\title{Evaluating functions as processes}
\author{Beniamino Accattoli\institute{Carnegie Mellon University - Pittsburgh, PA, US}}
\def\titlerunning{Evaluating functions as processes}
\def\authorrunning{B. Accattoli}

\begin{document}
\maketitle
\begin{abstract}
A famous result by Milner is that the $\l$-calculus can be simulated inside the $\pi$-calculus. This simulation, however, holds only modulo strong bisimilarity on processes, i.e. there is a slight mismatch between $\beta$-reduction and how it is simulated in the $\pi$-calculus. The idea is that evaluating a $\l$-term in the $\pi$-calculus is like running an environment-based abstract machine, rather than applying ordinary $\beta$-reduction. In this paper we show that such an abstract-machine evaluation corresponds to linear weak head reduction, a strategy arising from the representation of $\l$-terms as linear logic proof nets, and that the relation between the two is as tight as it can be. The study is also smoothly rephrased in the call-by-value case, introducing a call-by-value analogous of linear weak head reduction.
\end{abstract}

\input{introduction}

\input{ls-calculus}
\input{pi-calculus}
\input{cbn}
\input{cbv}

\section*{Conclusions}
We have shown how to refine the relation between the $\l$-calculus and the $\pi$-calculus, getting a perfect match of reductions steps in both call-by-name and call-by-value. The refinements crucially exploits rewriting rules at a distance, and unveil that the $\pi$-calculus evaluates $\l$-terms exactly as linear logic proof nets. A natural continuation  would be to extend these relations to calculi with multiplicities \cite{DBLP:journals/iandc/BoudolL96}, which are related to the study of observational equivalence. It would also be interesting to investigate linear weak applicative reduction, in particular in relation with complexity \cite{DBLP:conf/rta/AccattoliL12} or with Taylor-Ehrhard expansion \cite{DBLP:conf/csl/Ehrhard12}. Finally, given the compactness of the results and the involved reasoning about bound, free, and fresh variables, it would be interesting to try to formalize this work in Abella \cite{DBLP:conf/cade/Gacek08}, which is a proof assistant provided with a nominal quantifier precisely developed to cope with the $\pi$-calculus \cite{DBLP:journals/tocl/TiuM10} and where reasoning about untyped calculi with binders is very close to pen-and-paper reasoning \cite{DBLP:conf/cpp/Accattoli12}.


\input{main.bbl}
\end{document}

%% file: ben-macros-preamble.tex
\usepackage{ifthen}
\newboolean{talk}
\setboolean{talk}{false}
\newboolean{paper}
\setboolean{paper}{false}
\newboolean{french}
\setboolean{french}{false}
\newboolean{IEEEstyle}
\setboolean{IEEEstyle}{false}
\newboolean{lipicsstyle}
\setboolean{lipicsstyle}{false}
\newboolean{eptcsstyle}
\setboolean{eptcsstyle}{false}

%% file: ben-macros-paper.tex
\ifthenelse{\boolean{talk}}{}{\usepackage[utf8]{inputenc}}
\ifthenelse{\boolean{french}}{\usepackage[french]{babel}}{
  \ifthenelse{\boolean{lipicsstyle}}{}{\usepackage[english]{babel}}}
\usepackage{amssymb}
\usepackage{amsmath}
\usepackage{graphicx}
\usepackage{bussproofs}
\usepackage{fancybox}
\usepackage{cmll}
\usepackage{stmaryrd}
\usepackage{multirow}
\ifthenelse{\boolean{talk}}{\usepackage{color}}{
  \ifthenelse{\boolean{lipicsstyle}}{\usepackage{color}}{\usepackage[usenames]{color}}}

\ifthenelse{\boolean{IEEEstyle}}{
\input{thm-environments}}{}

\ifthenelse{\boolean{eptcsstyle}}{
\input{thm-environments}}{}

\newcommand{\tm}{t}
\newcommand{\tmtwo}{s}
\newcommand{\tmthree}{u}
\newcommand{\tmfour}{r}

\newcommand{\var}{x}
\newcommand{\vartwo}{y}
\newcommand{\varthree}{z}
\newcommand{\varfour}{w}

\newcommand{\vartwop}{y'}
\newcommand{\varthreep}{z'}

\newcommand{\varset}{\Delta}
\newcommand{\varsettwo}{\Gamma}
\newcommand{\varsetthree}{\Sigma}

\newcommand{\alpeq}{=_\alpha}

\newcommand{\rootRew}[1]{\mapsto_{#1}}
\newcommand{\Rew}[1]{\rightarrow_{#1}}

\renewcommand{\to}{\Rew{}}

\newcommand{\lssym}{{\tt ls}}

\newcommand{\db}{{\tt dB}}

\newcommand{\vsym}{{\tt v}}

\newcommand{\val}{v}

\newcommand{\ctxholep}[1]{[#1]}
\newcommand{\ctxhole}{\ctxholep{\cdot}}

\newcommand{\evvctx}[1]{E_{#1}}
\newcommand{\evvctxtwo}[1]{F_{#1}}

\newcommand{\evvctxp}[2]{\evvctx{#1}\ctxholep{#2}}
\newcommand{\evvctxptwo}[2]{\evvctxtwo{#1}\ctxholep{#2}}

\newcommand{\apctxv}[1]{A_{#1}}
\newcommand{\apctxvtwo}[1]{B_{#1}}

\newcommand{\apctxvp}[2]{\apctxv{#1}\ctxholep{#2}}
\newcommand{\apctxvptwo}[2]{\apctxvtwo{#1}\ctxholep{#2}}

\newcommand{\nbctx}{N}
\newcommand{\nbctxtwo}{M}
\newcommand{\nbvctx}[1]{\nbctx_{#1}}
\newcommand{\nbvctxtwo}[1]{\nbvctxtwo{#1}}
\newcommand{\nbvctxp}[2]{\nbvctx{#1}\ctxholep{#2}}
\newcommand{\nbvctxptwo}[2]{\nbvctxtwo{#1}\ctxholep{#2}}

\newcommand{\sctx}{L}
\newcommand{\sctxtwo}{L'}

\newcommand{\sctxv}[1]{\sctx_{#1}}
\newcommand{\sctxvtwo}[1]{\sctxtwo_{#1}}

\newcommand{\sctxvp}[2]{\sctxv{#1}\ctxholep{#2}}
\newcommand{\sctxvptwo}[2]{\sctxvtwo{#1}\ctxholep{#2}}

\newcommand{\defeq}{:=}
\newcommand{\grameq}{::=}

\newcommand{\isub}[2]{\{#1/#2\}}

\newcommand{\rtodb}{\rootRew{\db}}

\newcommand{\todbp}[1]{\Rew{\db#1}}
\newcommand{\todb}{\todbp{}}

\newcommand{\llsub}{\l_{lsub}}

\newcommand{\rtols}{\rootRew{\lssym}}

\newcommand{\tols}{\Rew{\lssym}}

\newcommand{\tohl}{\multimap}

\newcommand{\towhl}{\multimap}
\newcommand{\towhldb}{\multimap_\db}
\newcommand{\towhlls}{\multimap_\lssym}
\newcommand{\rtowhls}{\hat{\multimap}_\lssym}

\newcommand{\pr}{P}
\newcommand{\prtwo}{Q}
\newcommand{\prthree}{R}

\newcommand{\inp}[2]{#1(#2)}
\newcommand{\out}[2]{\overline{#1}\langle#2\rangle}
\newcommand{\paral}{\ |\ }

\newcommand{\topi}{\Rew{\pi}}

\newcommand{\lamtopia}[2]{\red{\llbracket \black{#1}\rrbracket}_{#2}}

\newcommand{\llbrace}{\{ \kern -0.27em \vert}
\newcommand{\rrbrace}{\vert \kern -0.27em \}}
\newcommand{\lamtopiv}[2]{\red{\llbrace \black{#1}\rrbrace}^{#2}}

\newcommand{\streq}{\equiv}
\newcommand{\fn}[1]{{\tt fn}(#1)}

\newcommand{\topim}{\Rew{\otimes}}

\newcommand{\topils}{\Rew{!}}

\newcommand{\topid}{\Rightarrow}
\newcommand{\rtopidm}{\rootRew{\otimes}}
\newcommand{\rtopidls}{\rootRew{!}}
\newcommand{\topidm}{\Rightarrow_{\otimes}}
\newcommand{\topidls}{\Rightarrow_{!}}

\newcommand{\grammarpipe}{\mathrel{\big |}}

\renewcommand{\l}{\lambda}
\newcommand{\ie}{{\em i.e.}}

\newcommand{\ih}{\textit{i.h.}}
\newcommand{\fv}[1]{{\tt fv}(#1)}

\newcommand{\deff}[1]{\textbf{#1}}
\newcommand{\boldemph}[1]{\textbf{#1}}

\newcommand{\red}[1]{{\color{red} {#1}}}
\newcommand{\blue}[1]{{\color{blue} {#1}}}

\newcommand{\black}[1]{{\color{black} {#1}}}

\newcommand{\ignore}[1]{}
\newcommand{\sep}{\hspace*{0.5cm}}

\newcommand{\mathintitle}[2]{\texorpdfstring{#1}{#2}}

\newcommand{\myinput}[1]{\ifthenelse{\boolean{withimages}}{\input{#1}}{}}

\newcommand{\gerard}{{G{\'e}rard}}

\newcommand{\linlogic}{linear logic}

\newcommand{\set}[1]{\{#1\}}
\newcommand{\disunion}{\uplus}

\newcommand{\pns}{proof nets}
\newcommand{\Pns}{Proof nets}

\newcommand{\lv}{\l_{\beta v}}

\newcommand{\lvsub}{\l_{vsub}}
\newcommand{\lvker}{\l_{vker}}

\newcommand{\tolwa}{\multimap_{\vsym}}
\newcommand{\tolwdb}{\multimap_{\vsym\db}}
\newcommand{\rtolwv}{\rootRew{\lssym\vsym}}
\newcommand{\tolwv}{\multimap_{\vsym\lssym}}



%% file: ben-macros-diagrams.tex
\usepackage{tikz}
\usetikzlibrary{matrix,arrows}
\usetikzlibrary{decorations.shapes}
\usetikzlibrary{decorations.text}
\usetikzlibrary{decorations.pathmorphing}
\usetikzlibrary{decorations.markings}

\tikzset{
node distance=1.3cm, auto,
every node/.style={font=\tiny },
ocenter/.style={baseline={([yshift=-.5ex, xshift=-.5ex]current bounding box)}},  
labelBeginAbove/.style={postaction={decorate,decoration={markings,mark=at position 0 with {\node[inner sep= 0.6pt, above=1pt]{\tiny #1};}} } },
labelBeginBelow/.style={postaction={decorate,decoration={markings,mark=at position 0 with {\node[inner sep= 0.6pt, below=1pt]{\tiny #1};}}}},
labelEndAbove/.style={postaction={decorate,decoration={markings,mark=at position 1 with {\node[inner sep= 0.6pt, above=1pt]{\tiny #1};}}}},
labelEndBelow/.style={postaction={decorate,decoration={markings,mark=at position 1 with {\node[inner sep= 0.6pt, below=1pt]{\tiny #1};}}}},
labelEndRight/.style={postaction={decorate,decoration={markings,mark=at position 1 with {\node[inner sep= 0.6pt, right=1pt]{\tiny #1};}}}},
labelEndLeft/.style={postaction={decorate,decoration={markings,mark=at position 1 with {\node[inner sep= 0.6pt, left=1pt]{\tiny #1};}}}}
}

%% file: local-macros.tex
\usepackage{amsthm}
\theoremstyle{plain}
    \newtheorem{theorem}{Theorem}
    \newtheorem{lemma}[theorem]{Lemma}

\theoremstyle{definition}
    \newtheorem{definition}[theorem]{Definition}
\theoremstyle{remark}

\newcommand{\ovar}{\red{a}}
\newcommand{\ovartwo}{\red{b}}
\newcommand{\ovarthree}{\red{c}}
\newcommand{\ovarfour}{\red{d}}
\newcommand{\ovarfive}{\red{e}}

\newcommand{\ltm}[1]{\tm^{\ovar}}
\newcommand{\ltmtwo}[1]{\tmtwo^{\ovartwo}}
\newcommand{\ltmthree}[1]{\tmthree^{\ovarthree}}
\newcommand{\ltmfour}[1]{\tmfour^{\ovarfour}}

\newcommand{\remph}[1]{\blue{#1}}

\newcommand{\specialchannels}{\red{A}}
\newcommand{\specialv}{special}

\renewcommand{\rtowhls}{\rootRew{\lssym}}

\renewcommand{\ctxholep}[1]{\remph{\llparenthesis} #1 \remph{\rrparenthesis}}

\renewcommand{\nbctx}{\remph{N}}
\renewcommand{\nbctxtwo}{\remph{M}}
\renewcommand{\nbvctx}[1]{\nbctx_{\remph{#1}}}
\renewcommand{\nbvctxtwo}[1]{\nbctxtwo_{\remph{#1}}}
\renewcommand{\nbvctxp}[2]{\nbvctx{#1}\ctxholep{#2}}
\renewcommand{\nbvctxptwo}[2]{\nbvctxtwo{#1}\ctxholep{#2}}

\renewcommand{\evvctx}[1]{\remph{E_{#1}}}
\renewcommand{\evvctxtwo}[1]{\remph{F_{#1}}}

\renewcommand{\apctxv}[1]{\remph{A_{#1}}}
\renewcommand{\apctxvtwo}[1]{\remph{B_{#1}}}

\renewcommand{\sctx}{\remph{L}}
\renewcommand{\sctxtwo}{\remph{L'}}

\renewcommand{\sctxv}[1]{\sctx_{\remph{#1}}}
\renewcommand{\sctxvtwo}[1]{\sctxtwo_{\remph{#1}}}

%% file: introduction.tex
\section*{Introduction}
A key result about the expressiveness of the $\pi$-calculus is that it can represent the $\l$-calculus, as it has been showed by Robin Milner \cite{DBLP:journals/mscs/Milner92}. During the nineties the relationship between the two systems has been explored in-depth, mostly by Davide Sangiorgi \cite{DBLP:journals/iandc/Sangiorgi94,DBLP:journals/mscs/Sangiorgi99} and \gerard\ Boudol \cite{DBLP:journals/iandc/BoudolL96,DBLP:journals/lisp/Boudol98}. Nowadays, it takes a relevant part in the standard reference for the $\pi$-calculus \cite{DBLP:books/daglib/0004377}, and in any introductory course about it. From the process calculus point of view, it helps in getting deeper insights into its theory, especially because the $\pi$-calculus is far less canonical then the $\l$-calculus. From the $\l$-calculus point of view, it provides new tools to analyze the behavior of $\l$-terms and the dynamics of $\beta$-reduction.
\begin{figure}
\centering
\input{mydiagrams}
\caption{Diagrams describing the relationship between terms and processes.\label{fig:diagrams}}
\end{figure}
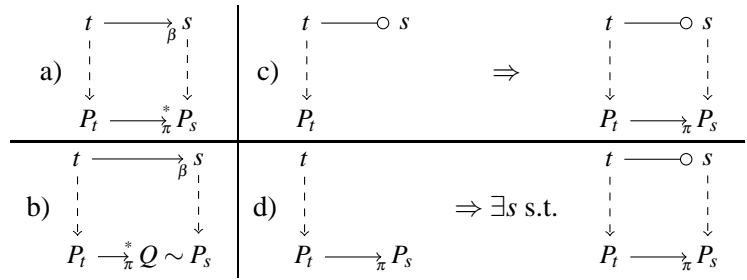

The idea is that the $\pi$-calculus can be considered as a sort of flexible abstract machine to which the $\l$-calculus can be compiled in various ways. There are in fact various encodings, each one corresponding to a particular evaluation strategy in the $\l$-calculus. In particular, Milner showed that Plotkin's call-by-name and call-by-value strategies \cite{DBLP:journals/tcs/Plotkin75} can be both faithfully represented.

The way in which the representation is \emph{faithful}, however, is quite subtle. It is looser than what one might expect, as the diagram in Figure \ref{fig:diagrams}.a \emph{does not hold}. It is only possible to get the diagram in Figure \ref{fig:diagrams}.b: $\pr_\tm$, the process representing $\tm$, does not reduce to $\pr_\tmtwo$, but to a process $\prtwo$ which is strongly bisimilar to $\pr_\tmtwo$. One might think that a better encoding could solve this problem, but this is a na\"ive expectation: the two systems compute in radically different ways, the mismatch is inherent. In Milner's result $\pr_\tmtwo$ and $\prtwo$ are strongly bisimilar, which means that they behave the same \emph{externally}, \ie\ in their interactions with every possible environment. However, the two processes behave in a quite different way \emph{internally}, \ie\ with respect to reductions. The discrepancy concerns the granularity of evaluation: $\l$-calculus uses a coarse, big-step substitution rule, while the $\pi$-calculus evaluates in small, fine-grained steps, as an 
abstract machine. Nonetheless, the evaluation of $\tm$ terminates if and only if the evaluation of the corresponding process $\pr_\tm$ terminates. In this sense, the representation is sometimes said to be sound and complete.

This paper refines the relationship between the $\l$-calculus and the $\pi$-calculus by extending the former with explicit substitutions---which may be considered as an alternative to abstract machines---in order to get a closer match of reduction steps. In the call-by-name case we show that the strategy corresponding to the evaluation in the $\pi$-calculus is exactly \emph{linear weak head reduction} $\tohl$, the small-step head strategy of linear logic proof nets \cite{DBLP:journals/tcs/MascariP94,DBLP:conf/rta/Accattoli12}. This notion of evaluation has connections with Krivine's abstract machine \cite{Danos04headlinear}, Bohm's separation theorem \cite{DBLP:journals/tcs/MascariP94}, computational complexity \cite{DBLP:conf/rta/AccattoliL12}, the geometry of interaction \cite{DBLP:journals/tcs/DanosR99}, game semantics \cite{DBLP:conf/lics/DanosHR96,DBLP:conf/fossacs/Clairambault11}, and the differential $\l$-calculus \cite{DBLP:conf/cie/EhrhardR06}. The relationship shown here is extremely strong. It is represented in the diagrams in Figure \ref{fig:diagrams}.c-d, which hold modulo structural equivalence only. They express the fact that the translation is a strong bisimulation \emph{with respect to reduction} (note that \emph{one} step maps to \emph{one} step, and vice-versa).

The relationship between the $\pi$-calculus and linear logic has been analyzed from various points of view \cite{DBLP:conf/elp/Miller92,DBLP:journals/tcs/Abramsky93,DBLP:journals/tcs/BellinS94,DBLP:journals/entcs/Beffara06,DBLP:journals/tcs/HondaL10,DBLP:journals/iandc/EhrhardL10,DBLP:conf/concur/CairesP10}. Our study essentially refines the work of Caires, Pfenning, and Toninho in \cite{DBLP:conf/fossacs/ToninhoCP12}, where the encodings of the $\l$-calculus in the $\pi$-calculus are re-understood as the encodings of $\l$-calculus into linear logic (due to Girard \cite{DBLP:journals/tcs/Girard87}, see also \cite{DBLP:journals/tcs/MaraistOTW99}). The refinement consists in looking to such encodings via linear logic proof nets, but replacing the explicit use of proof nets with the lighter and equivalent reformulations as calculi of explicit substitutions \emph{at a distance}, developed in \cite{DBLP:conf/csl/AccattoliG09,DBLP:conf/csl/AccattoliK10,phdaccattoli,DBLP:conf/flops/AccattoliP12,DBLP:conf/rta/Accattoli12,AccLSFA}.

\emph{Contributions}. In some sense there is not much original content in this paper. Damiano Mazza's master thesis \cite{mazzalinhead} (in French and unpublished) already developed the connection with linear weak head reduction. Similar ideas are sketched by Boudol in the introduction of \cite{DBLP:journals/lisp/Boudol98}. Also, Milner's seminal paper already suggested to use some environment device to refine the encodings, an idea that has then been explored by Vasconcelos \cite{DBLP:journals/jfp/Vasconcelos05} and recently by Cimini, Sacerdoti Coen, and Sangiorgi  \cite{DBLP:conf/tgc/CiminiCS10}. 

What is original here is the presentation. Our approach provides a remarkably compact development, confirming the relevance of explicit substitutions \emph{at a distance} as a very flexible syntactical tool. Our presentation simplifies in the extreme Mazza's study, by exploiting the simpler and more manageable reformulation of weak linear head reduction in the \emph{linear substitution calculus} \cite{DBLP:conf/rta/AccattoliL12,DBLP:conf/rta/Accattoli12}. In addition, by clarifying the connection with a crucial concept in the theory of linear logic, we get an important corollary \emph{for free}. In \cite{DBLP:conf/rta/AccattoliL12} it is proven that linear head reduction is at most quadratically longer than head reduction, and this result holds also with respect to the weak (\ie\ not under lambdas) variants of these reductions\footnote{The upper bound in \cite{DBLP:conf/rta/AccattoliL12} is exact, and it is based on a trasformation of reductions which applies to arbitrary 
reduction sequences, in particular even to non-terminating terms. For instance, the quadratic bound is reached by the evaluation of $(\l \var. \var\var)\l \var. \var\var$, which is weak.}. Plotkin's call-by-name strategy is the same thing as weak head reduction. Consequently, we get a quadratic relation between the call-by-name strategy and the evaluation in the $\pi$-calculus, which is a non-trivial quantitative refinement of Milner's result.

However, our contribution is not only about the presentation. The study of call-by-name is complemented by the study of a call-by-value encoding, from which we extract a call-by-value $\tolwa$ analogous of linear weak head reduction, which has never been considered before. We also show that this new strategy enjoys the analogous of the \emph{subterm property} \cite{DBLP:conf/rta/AccattoliL12} of linear weak head reduction, which is the basic property for complexity analysis. Last but not least, we give a presentation \emph{at a distance} of the rewriting rules of the $\pi$-calculus which is a contribution of independent interest.

Despite the compactness of the presentation, the details turned out to be quite delicate. The use of \emph{distance rules}, which are rewriting rules involving contexts (\ie\ terms with holes), is crucial. They reflect on terms the local rules of linear logic proof nets, and they are essential in order to get a strong bisimulation of reductions. These contexts can capture variables and names, a fact which requires a very careful analysis of the translations. This is why we present the proofs of the translation in details, almost certifying the result. Moreover, we use colors to ease the reading, so we suggest to read the paper simultaneously on paper and on a computer screen.

\emph{The relationship with \pns}. \Pns\ do not appear in this paper, we limit ourselves to the equivalent formulations as calculi at a distance. However, for the call-by-value calculus the detailed correspondence between terms and \pns\ can be found in \cite{AccLSFA} (which uses big-step rules, while here we use small-step rules), for call-by-name the interested reader may have a look to \cite{DBLP:conf/csl/AccattoliG09,phdaccattoli} (that do employ small-step rules, but in a slightly different way). On \pns, linear head reduction is the small step strategy which reduces only the cuts at level 0 which do not involve the auxiliary conclusions of $!$-boxes. The weak variant can be defined in exactly the same way if boxes are also used for $\parr$ (which in this context rather corresponds to the right rule for linear implication in intuitionistic \linlogic, and not to the $\parr$ of classical \linlogic). Using boxes for linear implication is less \emph{ad-hoc} than it may seem at first sight; a technical discussion of this issue is 
in Section 6 of \cite{AccLSFA}. This paper provides another justification for such boxes: they are needed to properly reflect evaluation in the $\pi$-calculus.

\emph{Plan of the paper.} Section \ref{s:ls-calculus} introduces the linear substitution calculus, and Section \ref{s:pi-calculus} introduces the presentation of the $\pi$-calculus that we use. Sections \ref{s:cbn} and \ref{s:cbv} study the call-by-name and the call-by-value encodings, respectively.

\emph{Acknowledgements.}
To Frank Pfenning, for having encouraged me to work out the details of this work, and to Damiano Mazza, for inspiration and comments on an early draft. This work was partially supported by the Qatar National Research Fund under grant NPRP 09-1107-1-168.

%
%
%
%
%

%% file: mydiagrams.tex
\begin{tabular}{c|cccccccccc}
a)
\begin{tikzpicture}[ocenter]
  \node (s) {\small $\tm$};
  \node (s1) [right of=s] {\small$\tmtwo$};
  \node (s2) [below of=s] {\small$\pr_\tm$};
  \node (t) [right of=s2] {\small$\pr_\tmtwo$};
  \draw[->, labelEndBelow=$\beta$] (s) to (s1);
  \draw[->,dashed] (s) to (s2);
  \draw[->, labelEndAbove=*, labelEndBelow=$\pi$] (s2) to  (t);
  \draw[->, dashed] (s1) to  (t);
\end{tikzpicture}
&
c)
\begin{tikzpicture}[ocenter]
  \node (s) {\small $\tm$};
  \node (s1) [right of=s] {\small$\tmtwo$};
  \node (s2) [below of=s] {\small$\pr_\tm$};
  \draw[-o] (s) to (s1);
  \draw[->,dashed] (s) to (s2);
\end{tikzpicture}
& $\Rightarrow$ &
\begin{tikzpicture}[ocenter]
  \node (s) {\small $\tm$};
  \node (s1) [right of=s] {\small$\tmtwo$};
  \node (s2) [below of=s] {\small$\pr_\tm$};
  \node (t) [right of=s2] {\small$\pr_\tmtwo$};
  \draw[-o] (s) to (s1);
  \draw[->,dashed] (s) to (s2);
  \draw[->, labelEndBelow=$\pi$] (s2) to  (t);
  \draw[->, dashed] (s1) to  (t);
\end{tikzpicture}\\\hline
b)
\begin{tikzpicture}[ocenter]
  \node (s) {\small $\tm$};  
  \node (s2) [below of=s] {\small$\pr_\tm$};
  \node (t) [right of=s2] {\small$\prtwo\sim\pr_\tmtwo$};
  \node (s1) [above of=t, right=3pt] {\small$\tmtwo$};
  \node at (t-|s1) [above =1pt](dummy) {};
  \draw[->, labelEndBelow=$\beta$] (s) to (s1);
  \draw[->,dashed] (s) to (s2);
  \draw[->, labelEndAbove=*, labelEndBelow=$\pi$] (s2) to  (t);
  \draw[->, dashed] (s1) to  (dummy);
\end{tikzpicture}

&
d)
\begin{tikzpicture}[ocenter]
  \node (s) {\small $\tm$};
  \node (s2) [below of=s] {\small$\pr_\tm$};
  \node (t) [right of=s2] {\small$\pr_\tmtwo$};
  \draw[->,dashed] (s) to (s2);
  \draw[->, labelEndBelow=$\pi$] (s2) to  (t);
\end{tikzpicture}
& $\Rightarrow \exists \tmtwo$ s.t. &
\begin{tikzpicture}[ocenter]
  \node (s) {\small $\tm$};  
  \node (s2) [below of=s] {\small$\pr_\tm$};
  \node (t) [right of=s2] {\small$\pr_\tmtwo$};
  \node (s1) [right of=s] {\small$\tmtwo$};
  \draw[-o] (s) to (s1);
  \draw[->,dashed] (s) to (s2);
  \draw[->, labelEndBelow=$\pi$] (s2) to  (t);
  \draw[->, dashed] (s1) to  (t);
\end{tikzpicture}

\end{tabular}

%% file: ls-calculus.tex
\section{The linear substitution calculus}
\label{s:ls-calculus}
The language of the \emph{linear substitution calculus} $\llsub$ is given by the following grammar for terms:
\begin{center}
 $\begin{array}{rcl}
\tm,\tmtwo,\tmthree,\tmfour&::=&  \var\mid \l \var. \tm \mid \tm \tmtwo\mid  \tm[\var/\tmtwo]
\end{array}$
\end{center}

\noindent The constructor $\tm[\var/\tmtwo]$ is called an \emph{explicit substitution} (of $\tmtwo$ for $\var$ in $\tm$, the usual (implicit) substitution is instead noted $\tm\isub{\var}{\tmtwo}$). Both $\l \var. \tm$ and $\tm[\var/\tmtwo]$ bind $\var$ in $\tm$. We are not going to define the full calculus (for which we refer to \cite{DBLP:conf/rta/AccattoliL12,DBLP:conf/rta/Accattoli12}), but only linear weak head reduction. However, let us point out that the linear substitution calculus is a variation over a calculus of explicit substitutions introduced by Robin Milner in \cite{DBLP:journals/entcs/Milner07}, to analyze the translation of $\l$-calculus to Bigraphs. 

 We shall use contexts extensively, so we define them formally. In particular, we need to specify the set $\varset$ of variables captured by a given context. A weak head context, or simply an \deff{evaluation context}, is a term of the following grammar (to ease the reading on screen all contexts will be in blue):
 \begin{center}
 $\begin{array}{rcl@{\hspace{3cm}}rcl}
\evvctx{\emptyset} &::=&  \ctxhole\mid \evvctx{\emptyset} \tm&
\evvctx{\varset\disunion\set{\var}} &::=& \evvctx{\varset}[\var/\tm] \mid \evvctx{\varset\disunion\set{\var}}\tm
\end{array}$
 \end{center}

\noindent A special case of evaluation context is given by \deff{substitution contexts}, noted $\sctxv{\varset}$ and defined by:
\begin{center}
$\begin{array}{rcl@{\hspace{3cm}}rcl}
\sctxv{\emptyset} &::=&  \ctxhole&
\sctxv{\varset\disunion\set{x}} &::=& \sctxv{\varset}[x/t]
\end{array}$
\end{center}

\begin{definition}
\deff{Linear weak head reduction} $\towhl$ is defined as the union of $\towhldb$ and $\towhlls$, which are given by the closure by evaluation contexts (\ie\ $\towhldb:=\evvctx{\varset}[\rtodb]$ and $\towhlls:=\evvctx{\varset}[\rtowhls]$) of the rules $\rtodb$ and $\rtowhls$ defined as:
\begin{center}
$\begin{array}{lll@{\hspace{2cm}}lll@{\hspace{.7cm}}l}
\sctxvp{\varset}{\l \var.\tm} \tmtwo & \rtodb &  \sctxvp{\varset}{\tm[\var/\tmtwo]}&
 \evvctxp{\varset}{\var}[\var/\tmtwo] &\rtowhls & \evvctxp{\varset}{\tmtwo}[\var/\tmtwo] & \mbox{with $\var\notin\varset$}\\
\end{array}$
\end{center}
\end{definition}
The rule $\rtowhls$ implicitly assumes the side-condition $\fv{\tmtwo}\cap\varset=\emptyset$. The assumption is implicit because it can always be guaranteed by $\alpha$-conversion: if $\tmthree=\evvctxp{\varset}{\var}[\var/\tmtwo]$ and $\fv{\tmtwo}\cap\varset\neq\emptyset$ then there exist a set of variables $\varsetthree$ and an evaluation context $\evvctxtwo{\varsetthree}$ s.t. $\tmthree\alpeq\evvctxptwo{\varsetthree}{\var}[\var/\tmtwo]$ and  $\fv{\tmtwo}\cap\varsetthree=\emptyset$. 

These rule are \emph{at a distance}, because their definition involves contexts, which is how locality on \pns\ is reflected on terms. In Milner's calculus the first rule does not use $\sctxvp{\varset}{\cdot}$. This is not a detail: the results in this paper would not hold with respect to Milner's original presentation.

It is natural to wonder in which sense the linear substitution calculus is \emph{linear}. In contrast to other linear calculi, variables may have multiple occurrences, and arguments are not forced to be used only once. A first superficial linear aspect of the calculus is that variable occurrences are substituted one at the time. A second much deeper aspect is that its head strategy---characterized by a factorization theorem in the same way as head reduction in $\l$-calculus \cite{DBLP:conf/rta/Accattoli12}---is \emph{linear head reduction}, whose main feature is the \emph{subterm property} (namely: any subterm $\tmthree$ which is duplicated at any point of a reduction $\tm \tohl^k\tmtwo$ is a subterm of $\tm$, whose size then does not depend on $k$) which implies that the implementation cost of \emph{every} step is linear (in the size of $\tm$, the parameter for complexity). This is a fundamental property, not enjoyed by any strategy in $\l$-calculus (for which the cost of one step is not even polynomial in the size of $\tm$), and which 
opens the way to the study of computational complexity \cite{DBLP:conf/rta/AccattoliL12}. Here we deal with linear \emph{weak} head reduction, which forbids reduction under abstractions. The restriction does not affect the subterm property.

%% file: pi-calculus.tex
\section{The \mathintitle{$\pi$}{pi}-calculus}
\label{s:pi-calculus}
The fragment of the $\pi$-calculus we use here is essentially the asynchronous calculus in \cite{DBLP:conf/csl/DeYoungCPT12} with both unary and binary inputs and outputs, morally corresponding to the exponential and the multiplicative connectives of linear logic (in the typed case of \cite{DBLP:conf/csl/DeYoungCPT12}) and without sums (which correspond to the additives). The only change is that we do not use their forwarding processes\footnote{Forwarding processes correspond to axioms in \linlogic. In terms of \pns, avoiding forwarding processes correspond to use an interaction nets presentation, \ie\ to work modulo cut-elimination on axioms.}. The grammar is:
\begin{center}
 $\begin{array}{rcl}
 \pr,\prtwo,\prthree&::=&  0\grammarpipe \out{\var}{\vartwo} \grammarpipe  \out{\var}{\vartwo,\varthree} \grammarpipe   \nu \var\pr \grammarpipe \inp{\var}{\vartwo,\varthree}.\pr \grammarpipe! \inp{\var}{\vartwo}.\pr \grammarpipe  \pr \paral \prtwo
 \end{array}$
\end{center}

\noindent We need a notion of context also for processes. A \deff{non-blocking context} is given by:
\begin{center}
$\begin{array}{rcl@{\sep\sep\sep\sep\sep}rcl}
\nbvctx{\emptyset}&::=&  \ctxhole\grammarpipe  \nbvctx{\emptyset}\paral Q \grammarpipe  P \paral \nbvctx{\emptyset} &
\nbvctx{\varset\disunion x}&::=&  \nu \var\nbvctx{\varset} \grammarpipe\nbvctxp{\emptyset}{\nbvctx{\varset\disunion x}}
\end{array}$
\end{center}

\noindent The language is considered modulo \deff{structural congruence}, \ie\ the minimum equivalence relation generated by the following rules and closed by non-blocking contexts:
\begin{center}
$\begin{array}{c@{\sep}c@{\sep}c@{\sep}c@{\sep}c@{\sep}c@{\sep}c}
\AxiomC{}
\UnaryInfC{$P\paral 0  \streq P$}
\DisplayProof	
& 
\AxiomC{}
\UnaryInfC{$P\paral (Q\paral R)\streq (P\paral Q)\paral R$}
\DisplayProof	
& 
\AxiomC{}
\UnaryInfC{$P\paral Q\streq Q\paral P$}
\DisplayProof	
\\\\
\AxiomC{}
\UnaryInfC{$\nu \var 0\streq 0$}
\DisplayProof	
& 
\AxiomC{$x\notin\fn{P}$}
\UnaryInfC{$\pr\paral \nu \var \prtwo\streq \nu \var. (\pr\paral \prtwo)$}
\DisplayProof	
& 
\AxiomC{}
\UnaryInfC{$\nu \var\nu \vartwo\pr\streq \nu \vartwo\nu\var\pr$}
\DisplayProof	
\end{array}$
\end{center}

In order to prove the simulation theorems we will use the following three properties of $\streq$, proved by easy inductions on $\nbvctx{\varset}$, $\pr$, and $\nbvctx{\varset}$, respectively (the set of free variables of a context is defined as for processes but using $\fn{\ctxhole}=\emptyset$).

\begin{lemma}
\label{l:pi-der-rules}
Let $\varset$ be a set of variables, $\nbvctx{\varset}$ a non-blocking context, $\pr$ a process s.t. $\fn{\pr}\cap\varset=\emptyset$, and $\var,\vartwo\notin\varset$. Then:
\begin{enumerate}
\item \label{p:pi-der-rules-one}$\nbvctxp{\varset}{\prtwo}\paral \pr\streq \nbvctxp{\varset}{\prtwo\paral \pr}$.
 \item \label{p:pi-der-rules-three}If $\var\notin\fn{\pr}$ then $\nu\var \pr\streq \pr$.
 \item \label{p:pi-der-rules-four}If $\var\notin\fn{\nbvctx{\varset}}$ then $\nu\var \nbvctxp{\varset}{\pr}\streq  \nbvctxp{\varset}{\nu\var\pr}$.
\end{enumerate}

\end{lemma}

The rewriting rules are the following:
\begin{center}
$\begin{array}{lll@{\sep\sep\sep\sep}ccc}
\out{\var}{\vartwo,\varthree}\paral \inp{\var}{\vartwop,\varthree'}.\prtwo &\topim &\prtwo\isub{\vartwop}{\vartwo}\isub{\varthreep}{\varthree}&
\out{\var}{\vartwo}\paral !\inp{\var}{\varthree}.\prtwo &\topils & \prtwo\isub{\varthree}{\vartwo}\paral !\inp{\var}{\varthree}.\prtwo
\end{array}$
\end{center}

\noindent as usual they are both closed by non-blocking contexts and considered modulo $\equiv$. The second rule puts together replication and unary communication as in \cite{DBLP:conf/fossacs/ToninhoCP12,DBLP:conf/csl/DeYoungCPT12}.

\boldemph{$\pi$-calculus, at a distance.}
In order to simplify the proof of the bisimulation, we are going to use an alternative but equivalent definition of reduction in the $\pi$-calculus. Essentially, we have to reformulate the $\pi$-calculus \emph{at a distance}. The use of the structural equivalence in the definition of the rewriting relation of the $\pi$-calculus induces some annoying complications when one tries to reflect process reductions on terms. We are going to reformulate the reduction rules via non-blocking contexts, and get rid of structural equivalence.

The rewriting rules $\topidm$ and $\topidls$ are given by the closure by non-blocking contexts (but are not closed by structural congruence) of the following relations: if $\var\notin\varset\cup\varsettwo$ then
\begin{center}
$\begin{array}{l@{\sep}l@{\sep}l@{\sep}c@{\sep}c@{\sep}c}
\nbvctxp{\varset}{\out{\var}{\vartwo,\varthree}}\paral \nbvctxptwo{\varsettwo}{\inp{\var}{\vartwop,\varthreep}.\pr} 
&\rtopidm 
&\nbvctxptwo{\varsettwo}{\nbvctxp{\varset}{\pr\isub{\vartwop}{\vartwo}\isub{\varthreep}{\varthree}}}\\

\nbvctxp{\varset}{\out{\var}{\vartwo}}\paral \nbvctxptwo{\varsettwo}{!\inp{\var}{\varthree}.\pr} 
&\rtopidls 
&\nbvctxptwo{\varsettwo}{\nbvctxp{\varset}{ \pr\isub{\varthree}{\vartwo}\paral!\inp{\var}{\varthree}.\pr}}\\
\end{array}$
\end{center}

\noindent Actually, one should ask three futher conditions on variables: 1) $\varset\cap\varsettwo=\emptyset$; 2) $\varset\cap\fv{\pr}=\emptyset$; 3) $\fv{\nbvctx{\varset}}\cap\varsettwo=\emptyset$. It is easily seen, however, that these conditions can always be satisfied by choosing an $\alpha$-equivalent term, as it is the case for the $\rtols$ rule of $\llsub$.
Essentially, these rules re-formulate as reduction rules the $\tau$-transitions of the alternative presentation of the $\pi$-calculus as a labeled transition system, which is used to study the interaction of a process with its environment. Here, the new rules are more convenient than labeled transitions, because on $\l$-terms there is no analogous of the transitions whose label is not $\tau$ (and $\tau$-transitions are defined using the non-$\tau$ transitions). This reformulation is justified by the following lemma, whose proof is along the one of the harmony lemma in \cite{DBLP:books/daglib/0004377} (p. 51).

\begin{lemma}
\label{l:harmony}
\hfill
\begin{enumerate}
\item \boldemph{$\streq$ is a strong bisimulation with respect to $\topid$}: $\pr\streq\topidm\prtwo$ iff $\pr\topidm\streq\prtwo$, and $\pr\streq\topidls\prtwo$ iff $\pr\topidls\streq\prtwo$.
\item \boldemph{Harmony of $\topid$ and $\topi$}: \label{p:harmony-one}$\pr\topim\prtwo$ iff $\pr\topidm\streq\prtwo$, and $\pr\topils\prtwo$ iff $\pr\topidls\streq\prtwo$.
\end{enumerate}
\end{lemma}

Curiously, the first formulation of the $\pi$-calculus was as a labeled transition system; the notions of reduction and structural congruence were introduced by Milner only later on, to study the relationship with the $\l$-calculus \cite{DBLP:journals/mscs/Milner92}. Our formulation at a distance of the $\pi$-calculus---motivated in exactly the same way---is a contribution of independent interest, probably the main one from the $\pi$-calculus point of view. It also shows that distance rules are a general syntactic principle whose relevance extends beyond explicit substitutions.
%
%
%
%
%
%
%
%

%% file: cbn.tex
\section{The call-by-name encoding}
\label{s:cbn}
As for the ordinary $\l$-calculus, the translation from $\llsub$ to the $\pi$-calculus is parametrized by a special channel name $\ovar$. Actually, we assume that these \deff{\specialv\ channel names} are taken from a set $\specialchannels$ which is disjoint from the set of variable names, and whose elements are denoted $\ovar,\ovartwo,\ovarthree,\ovarfour,\ldots$.

The translation is given by (on screen it is in red):
\begin{center}
$\begin{array}{rcl@{\sep\sep\sep}rcl@{\sep}ll}
\lamtopia{\var}{\ovar} &\defeq&  \out{\var}{\ovar}&
\lamtopia{\tm\tmtwo}{\ovar} &\defeq& \nu \ovartwo\nu \var  (\lamtopia{\tm}{\ovartwo}\paral \out{\ovartwo}{\var,\ovar}\paral!\inp{\var}{\ovarthree}.\lamtopia{\tmtwo}{\ovarthree}) & \var \mbox{ is fresh}\\

\lamtopia{\l \var. \tm}{\ovar} &\defeq&\inp{\ovar}{\var, \ovartwo}.\lamtopia{\tm}{\ovartwo}&

\lamtopia{\tm[\var/\tmtwo]}{\ovar} &\defeq&\nu \var (\lamtopia{\tm}{\ovar}\paral !\inp{\var}{\ovartwo}.\lamtopia{\tmtwo}{\ovartwo})\\
\end{array}$
\end{center}

\noindent Modulo minor details, this is the original call-by-name encoding given by Milner. With respect to the relation with linear logic developed in \cite{DBLP:conf/csl/DeYoungCPT12}, \specialv\ names correspond exactly to multiplicative formulas, while variable names correspond to exponential formulas. 

An easy induction on the translation shows:

\begin{lemma}
\label{l:transl-free-names}
Let $\tm$ be a term. Then $\fn{\lamtopia{\tm}{\ovar}}=\fv{\tm}\disunion\set{\ovar}$.
\end{lemma}

To relate terms and processes we need to prove a property of the translation, concerning its action on contexts: it maps evaluation contexts to non-guarding contexts of a special form.

\begin{lemma}[Relating $\evvctx{}$ and $\nbvctx{}$ via $\lamtopia{\cdot}{\ovar}$]
\label{l:wh-to-nb}
Let $\varset$ be a set of variable names, $\evvctx{\varset}$ an evaluation context, and $\ovar$ a \specialv\ name. There exist a set of names $\varsettwo$ (possibly containing both variables and \specialv\ names), a non-blocking context $\nbvctx{\varset\disunion\varsettwo}$ and a \specialv\ name $\ovartwo$ s.t. $\lamtopia{\evvctxp{\varset}{\tm}}{\ovar}= \nbvctxp{\varset\disunion\varsettwo}{\lamtopia{\tm}{\ovartwo}}$ and $\varsettwo\cap\fv{\tm}=\emptyset$ for every term $\tm$. Moreover, if $\evvctx{\varset}$ is a substitution context $\sctxv{\varset}$ then $\ovar=\ovartwo$, $\varsettwo=\emptyset$, and $\nbvctx{\varset}$ does not depend on $\ovar$.
\end{lemma}

\proof
By induction on $\evvctx{\varset}$. The base case is given by the empty context $\evvctx{\emptyset}=\ctxhole$, and it is trivial, just take $\varsettwo:=\emptyset$, $\nbvctx{\emptyset}:=\ctxhole$, and $\ovartwo=\ovar$. The inductive cases:
\begin{itemize}
\item \emph{Left of an application}, $\evvctx{\varset}=\evvctxtwo{\varset} \tmtwo$: if $\var$ is a fresh variable name:
\begin{center}
$\begin{array}{llllllll}
\lamtopia{\evvctxp{\varset}{\tm}}{\ovar} &=&\lamtopia{\evvctxptwo{\varset}{\tm} \tmtwo}{\ovar} 
&=& \nu \ovarfour\nu \var  (\lamtopia{\evvctxptwo{\varset}{\tm}}{\ovarfour}\paral \out{\ovarfour}{\var,\ovar}\paral!\inp{\var}{\ovarthree}.\lamtopia{\tmtwo}{\ovarthree})\\
&&&=_{\ih}& \nu \ovarfour\nu \var  (\nbvctxptwo{\varset\disunion\varsetthree}{\lamtopia{\tm}{\ovartwo}} \paral \out{\ovarfour}{\var,\ovar}\paral!\inp{\var}{\ovarthree}.\lamtopia{\tmtwo}{\ovarthree})
&=& \nbvctxp{
\varset\disunion\varsetthree\disunion\set{\ovarfour,\var}}{\lamtopia{\tm}{\ovartwo}}
\end{array}$
\end{center}

\noindent By \ih\ we get that $\varsetthree\cap\fv{\tm}=\emptyset$. By definition of the translation $\var$ is fresh, so $x\notin \fv{\tm}$. We then conclude by taking $\varsettwo:=\varsetthree\disunion\set{\ovarfour,\var}$.

\item \emph{Left of a substitution}, $\evvctx{\varset\disunion\set{\var}}=\evvctxtwo{\varset} [\var/\tmtwo]$:
\begin{center}
$\begin{array}{llllllllll}
\lamtopia{\evvctxp{\varset\disunion\set{\var}}{\tm}}{\ovar} &=&\lamtopia{\evvctxptwo{\varset}{\tm} [\var/\tmtwo]}{\ovar} 
&=&\nu \var (\lamtopia{\evvctxptwo{\varset}{\tm}}{\ovar}\paral !\inp{\var}{\ovarthree}.\lamtopia{\tmtwo}{\ovarthree})\\
&&&=_{\ih}&\nu \var (\nbvctxptwo{\varset\disunion\varsettwo}{\lamtopia{\tm}{\ovartwo}}\paral !\inp{\var}{\ovarthree}.\lamtopia{\tmtwo}{\ovarthree})
&=& \nbvctxp{\varset\disunion\set{\var}\disunion\varsettwo}{\lamtopia{\tm}{\ovartwo}}
\end{array}$
\end{center}

\noindent and the \ih\ also gives $\varsettwo\cap\fv{\tm}=\emptyset$.

Now suppose that $\evvctx{\varset\disunion\set{\var}}$ (and thus $\evvctxtwo{\varset}$) is a substitution context $\sctxv{\varset}$. Then by \ih\ we get $\nbvctxtwo{\varset}$ not depending on $\ovar$ s.t.:
\begin{center}
$\begin{array}{lllllllll}
\lamtopia{\evvctxp{\varset\disunion\set{\var}}{\tm}}{\ovar} &=&\lamtopia{\evvctxptwo{\varset}{\tm} [\var/\tmtwo]}{\ovar} 
&=&\nu \var (\lamtopia{\evvctxptwo{\varset}{\tm}}{\ovar}\paral !\inp{\var}{\ovarthree}.\lamtopia{\tmtwo}{\ovarthree})\\
&&&=_{\ih}&\nu \var (\nbvctxptwo{\varset}{\lamtopia{\tm}{\ovar}}\paral !\inp{\var}{\ovarthree}.\lamtopia{\tmtwo}{\ovarthree})
&=& \nbvctxp{\varset\disunion\set{\var}}{\lamtopia{\tm}{\ovar}}
\end{array}$
\end{center}

\noindent Where clearly $\nbvctx{\varset\disunion\set{\var}}$ does not depend on $\ovar$.
\qed
\end{itemize}

We can now proceed with the simulation.

\begin{theorem}[$\topi$ strongly simulates $\towhl$ via $\lamtopia{\cdot}{\ovar}$]
\label{tm:simulation}
\hfill
\begin{enumerate}
\item \label{p:simulation-one}$\tm\towhldb \tmtwo$ implies $\lamtopia{\tm}{\ovar} \topidm \streq\lamtopia{\tmtwo}{\ovar} $.
\item \label{p:simulation-two} $\tm\towhlls \tmtwo$ implies $\lamtopia{\tm}{\ovar} \topidls\streq\lamtopia{\tmtwo}{\ovar} $.
\end{enumerate}
\end{theorem}

\proof
1. Two cases:
\begin{itemize}
 \item \emph{Root rewriting step}: first without $\sctxvp{\varset}{\cdot}$: $(\l x.M) N \rtodb  M[x/N]$
 \begin{center}
$\begin{array}{llllllllll}
\lamtopia{(\l \var. \tm) \tmtwo}{\ovar} &=&  \nu \ovartwo\nu \vartwo  (\lamtopia{\l \var. \tm}{\ovartwo}\paral \out{\ovartwo}{\vartwo,\ovar}\paral!\inp{\vartwo}{\ovarthree}.\lamtopia{\tmtwo}{\ovarthree}) 
&=& \nu \ovartwo\nu \vartwo  (\inp{\ovartwo}{\var,\ovarfive}.\lamtopia{\tm}{\ovarfive}\paral \out{\ovartwo}{\vartwo,\ovar}\paral!\inp{\vartwo}{\ovarthree}.\lamtopia{\tmtwo}{\ovarthree}) \\
&\topidm& \nu \ovartwo \nu \vartwo  (\lamtopia{\tm}{\ovar}\isub{\var}{\vartwo}\paral !\inp{\vartwo}{\ovarthree}.\lamtopia{\tmtwo}{\ovarthree}) 
&\alpeq&\nu \ovartwo \nu \var  (\lamtopia{\tm}{\ovar}\paral !\inp{\var}{\ovarthree}.\lamtopia{\tmtwo}{\ovarthree}) \\
&=&\nu \ovartwo \lamtopia{\tm[\var/\tmtwo]}{\ovar}
&\streq&\lamtopia{\tm[\var/\tmtwo]}{\ovar}
\end{array}$
\end{center}

\noindent The $\alpeq$-step is justified by the fact that $\vartwo$ is introduced fresh in the first line. The $\streq$ step is justified by Lemma \ref{l:transl-free-names}, for which the only free \specialv\  name occurring in $\lamtopia{\tm}{\ovar}$ is $\ovar$, and by Lemma \ref{l:pi-der-rules}.\ref{p:pi-der-rules-three}, which allow us to remove the useless $\nu \ovartwo$.

Now, if $\sctxvp{\varset}{\l \var.\tm} \tmtwo \rtodb  \sctxvp{\varset}{\tm[\var/\tmtwo]}$ we get (some explanations follow):
\begin{center}
$\begin{array}{rcl}
\lamtopia{\sctxvp{\varset}{\l \var. \tm} \tmtwo}{\ovar} &=&  \nu \ovartwo \nu \vartwo  (\lamtopia{\sctxvp{\varset}{\l \var. \tm}}{\ovartwo}\paral \out{\ovartwo}{\vartwo,\ovar}\paral!\inp{\vartwo}{\ovarthree}.\lamtopia{\tmtwo}{\ovarthree}) \\
 &=_{Lem. \ref{l:wh-to-nb}}& \nu \ovartwo \nu \vartwo  (\nbvctxp{\varset}{\lamtopia{\l \var. \tm}{\ovartwo}}\paral \out{\ovartwo}{\vartwo,\ovar}\paral!\inp{\vartwo}{\ovarthree}.\lamtopia{\tmtwo}{\ovarthree}) \\

 &=& \nu \ovartwo \nu \vartwo  (\nbvctxp{\varset}{\inp{\ovartwo}{\var,\ovarfive}.\lamtopia{\tm}{\ovarfive}}\paral \out{\ovartwo}{\vartwo,\ovar}\paral!\inp{\vartwo}{\ovarthree}.\lamtopia{\tmtwo}{\ovarthree}) \\
 
 &\topidm& \nu \ovartwo \nu \vartwo ( \nbvctxp{\varset}{\lamtopia{\tm}{\ovar}\isub{\var}{\vartwo}\isub{\ovarfive}{\ovar}}\paral!\inp{\vartwo}{\ovarthree}.\lamtopia{\tmtwo}{\ovarthree} )\\
  &\alpeq& \nu \ovartwo \nu \var  (\nbvctxp{\varset}{\lamtopia{\tm}{\ovar}}\paral!\inp{\var}{\ovarthree}.\lamtopia{\tmtwo}{\ovarthree}) \\
  &\streq_{Lem. \ref{l:pi-der-rules}.\ref{p:pi-der-rules-one}\& Lem. \ref{l:pi-der-rules}.\ref{p:pi-der-rules-four}}& \nu \ovartwo \nbvctxp{\varset}{\nu \var(\lamtopia{\tm}{\ovar}\paral!\inp{\var}{\ovarthree}.\lamtopia{\tmtwo}{\ovarthree})} \\

 &=& \nu \ovartwo  \nbvctxp{\varset}{\lamtopia{\tm[\var/\tmtwo]}{\ovar}} \\
 &=_{Lem. \ref{l:wh-to-nb}}& \nu \ovartwo  \lamtopia{\sctxv{\varset}[\tm[\var/\tmtwo]]}{\ovar} \\
 &\streq_{Lem. \ref{l:transl-free-names} \& Lem. \ref{l:pi-der-rules}.\ref{p:pi-der-rules-three}}& \lamtopia{\sctxv{\varset}[\tm[\var/\tmtwo]]}{\ovar} \\

\end{array}$
\end{center}

\noindent The $\alpeq$-step and the last step are justified as before. In the first application of $\streq$ we can apply Lemma \ref{l:pi-der-rules}.\ref{p:pi-der-rules-one} because by hypothesis $\var\notin\varset$ and $\fv{\tmtwo}\cap\varset=\emptyset$, and Lemma \ref{l:pi-der-rules}.\ref{p:pi-der-rules-four} because $\var\notin\fn{\nbvctx{\varset}}$. The two applications of Lemma \ref{l:wh-to-nb} are with respect to different \specialv\ names $\ovar$ and $\ovartwo$, but this is sound: the \emph{moreover} part of Lemma \ref{l:wh-to-nb} guarantees that in the case of a substitution context $\sctxv{\varset}$ the corresponding context $\nbvctx{\varset}$ does not depend on the name.

\item \emph{Inductive step}:  $\evvctxp{\varset}{\tm}\todb \evvctxp{\varset}{\tmtwo}$ because $\tm\rtodb \tmtwo$. Let us recall that by definitions reductions in the $\pi$-calculus are closed by non-blocking contexts. Then:
\begin{center}
$\begin{array}{llllllllll}
\lamtopia{\evvctxp{\varset}{\tm}}{\ovar} 
&=_{Lem. \ref{l:wh-to-nb}} & 
\nbvctxp{\varset\disunion\varsettwo}{\lamtopia{\tm}{\ovartwo}}
&\topidm&
\nbvctxp{\varset\disunion\varsettwo}{\lamtopia{\tmtwo}{\ovartwo}}
&=_{Lem. \ref{l:wh-to-nb}} & 
\lamtopia{\evvctxp{\varset}{\tmtwo}}{\ovar} 
\end{array}$
\end{center}
\end{itemize}

\noindent 2. For $\tols$ the inductive case is as for $\todb$. The base case is $\evvctxp{\varset}{\var}[\var/\tmtwo]\towhlls \evvctxp{\varset}{\tmtwo}[\var/\tmtwo]$ with $\var\notin\varset$:
\begin{center}
$\begin{array}{lllllll}
\lamtopia{\evvctxp{\varset}{\var}[\var/\tmtwo]}{\ovar} 
&=&\nu \var (\lamtopia{\evvctxp{\varset}{\var}}{\ovar}\paral !\inp{\var}{\ovartwo}.\lamtopia{\tmtwo}{\ovartwo})
&=_{Lem. \ref{l:wh-to-nb}} & 
\nu \var (\nbvctxp{\varset\disunion\varsettwo}{\lamtopia{\var}{\ovarthree}}\paral !\inp{\var}{\ovartwo}.\lamtopia{\tmtwo}{\ovartwo})\\
&=&\nu \var (\nbvctxp{\varset\disunion\varsettwo}{\out{\var}{\ovarthree}}\paral !\inp{\var}{\ovartwo}.\lamtopia{\tmtwo}{\ovartwo})

&\topidls&\nu \var \nbvctxp{\varset\disunion\varsettwo}{\lamtopia{\tmtwo}{\ovarthree}\paral !\inp{\var}{\ovartwo}.\lamtopia{\tmtwo}{\ovartwo}}\\

&\streq_{Lem. \ref{l:pi-der-rules}.\ref{p:pi-der-rules-one}} &\nu \var (\nbvctxp{\varset\disunion\varsettwo}{\lamtopia{\tmtwo}{\ovarthree}}\paral !\inp{\var}{\ovartwo}.\lamtopia{\tmtwo}{\ovartwo})

&=_{Lem. \ref{l:wh-to-nb}}&\nu \var (\lamtopia{\evvctxp{\varset}{\tmtwo}}{\ovar}\paral !\inp{\var}{\ovartwo}.\lamtopia{\tmtwo}{\ovartwo})\\
&=& \lamtopia{\evvctxp{\varset}{\tmtwo}[\var/\tmtwo]}{\ovar} 
\end{array}$
\end{center}

\noindent where the $\streq$-step is justified by the fact that by hypothesis and by Lemma \ref{l:wh-to-nb} ($\var\notin\varsettwo$) we get that $(\fv{s}\disunion\set{\var,\ovartwo})\cap (\Delta\disunion\Gamma)=\emptyset$, and so we can apply Lemma \ref{l:pi-der-rules}.\ref{p:pi-der-rules-one}.
\qed

\paragraph{The converse relation.} To simulate process reductions on $\l$-terms we need a lemma, which is a converse to Lemma \ref{l:wh-to-nb}.

\begin{lemma}
\label{l:conv-ctx}
Let $\varset$ and $\varsettwo$ be a set of variable names and a set of \specialv\ names, respectively.
\begin{enumerate}
\item \label{p:conv-ctx-one} If $\lamtopia{\tm}{\ovar}=\nbvctxp{\varset\disunion\varsettwo}{\inp{\ovar}{\vartwo,\ovartwo}.\pr}$ with $\ovar\notin\varsettwo$ then $\Gamma=\emptyset$ and exist $\tmtwo$ and $\sctxv{\varset}$ s.t. $\pr=\lamtopia{\tmtwo}{\ovartwo}$ and $\tm=\sctxvp{\varset}{\l \vartwo.\tmtwo}$.
\item \label{p:conv-ctx-two} If $\lamtopia{\tm}{\ovar}=\nbvctxp{\varset\disunion\varsettwo}{\out{\var}{\ovarthree}}$ with $\var\notin\varset$ then exist $\varsetthree\subseteq\varset$ and $\evvctx{\varsetthree}$ s.t. $\tm=\evvctxp{\varsetthree}{\var}$ (and $\var\notin\varsetthree$).
\end{enumerate}
\end{lemma}

\proof
Both points are by induction on $\tm$:
\begin{itemize}
\item \emph{Variable}: 
\begin{enumerate}
\item The hypothesis is false and there is nothing to prove.
\item By definition of $\lamtopia{\cdot}{\ovar}$, taking the empty context (and $\varset=\emptyset$).
\end{enumerate}
\item \emph{Abstraction}: 
\begin{enumerate}
\item By definition of $\lamtopia{\cdot}{\ovar}$, taking the empty context (and $\varset=\emptyset$).
\item The hypothesis is false and there is nothing to prove.
\end{enumerate}

\item \emph{Application}: if $\tm=\tmthree\tmfour$ then $\lamtopia{\tmthree\tmfour}{\ovar}=\nu \ovartwo\nu \varthree  (\lamtopia{\tmthree}{\ovartwo}\paral \out{\ovartwo}{\varthree,\ovar}\paral!\inp{\varthree}{\ovarthree}.\lamtopia{\tmfour}{\ovarthree})$ with $\varthree$ fresh. 
\begin{enumerate}
\item By Lemma \ref{l:transl-free-names} $\ovar\notin\fn{\lamtopia{\tmthree}{\ovartwo}}$, and so there is no context $\nbvctx{\varset\disunion\varsettwo}$ s. t. $\lamtopia{\tm}{\ovar}=\nbvctxp{\varset\disunion\varsettwo}{\inp{\ovar}{\vartwo,\ovartwo}.\pr}$, hence the hypothesis is false and there is nothing to prove.

\item It must be that $\lamtopia{\tmthree}{\ovar}=\nbvctxptwo{\varset'\disunion\varsettwo'}{\out{\var}{\ovarthree}}$ with $\varset=\varset'\disunion\set{\varthree}$ and $\varsettwo=\varsettwo'\disunion\set{\ovar}$. Then by \ih\ there exist $\varsetthree\subseteq\varset'$ and $\evvctxtwo{\varsetthree}$ s.t. $\tmthree=\evvctxptwo{\varsetthree}{\var}$. We conclude taking $\evvctx{\varsetthree}:=\evvctxtwo{\varsetthree}\tmfour$.
\end{enumerate}

\item \emph{Substitution}: if $\tm=\tmthree[\varthree/\tmfour]$ then $\lamtopia{\tmthree[\varthree/\tmfour]}{\ovar}= \nu \varthree (\lamtopia{\tmthree}{\ovar}\paral !\inp{\varthree}{\ovartwo}.\lamtopia{\tmfour}{\ovartwo})$.
\begin{enumerate}
\item  If $\lamtopia{\tm}{\ovar}=\nbvctxp{\varset\disunion\varsettwo}{\inp{\ovar}{\vartwo,\ovartwo}.\pr}$ then it must be that exists $\nbvctxptwo{\varset'\disunion\varsettwo}{\cdot}$  with $\varset=\varset'\disunion\set{\varthree}$ s.t. $\lamtopia{\tmthree}{\ovartwo}=\nbvctxptwo{\varset'\disunion\varsettwo}{\inp{\ovar}{\vartwo,\ovartwo}.\pr}$ and $\nbvctx{\varset\disunion\varsettwo}=\nu \varthree (\nbvctxtwo{\varset'\disunion\varsettwo}\paral !\inp{\varthree}{\ovartwo}.\lamtopia{\tmfour}{\ovartwo})$. By \ih\ we get $\Gamma=\emptyset$, $\tmthree=\sctxvptwo{\varset'}{\l \vartwo.\tmtwo}$, and $\pr=\lamtopia{\tmtwo}{\ovartwo}$. We conclude taking $\sctxv{\varset}:=\sctxvtwo{\varset'}[\varthree/\tmfour]$.

\item It must be that $\lamtopia{\tmthree}{\ovar}=\nbvctxptwo{\varset'\disunion\varsettwo'}{\out{\var}{\ovarthree}}$ with $\varset=\varset'\disunion\set{\varthree}$ and $\varsettwo=\varsettwo'\disunion\set{\ovar}$. Then by \ih\ there exist $\varsetthree'\subseteq\varset'$ and $\evvctxtwo{\varsetthree'}$ s.t. $\tmthree=\evvctxptwo{\varsetthree'}{\var}$. We conclude taking $\varsetthree:=\varsetthree'\disunion\set{\varthree}$ and $\evvctx{\varsetthree}:=\evvctxtwo{\varsetthree'}[\varthree/\tmfour]$.
\qed
\end{enumerate}
\end{itemize}

Now, we can prove that any process reduction from $\lamtopia{\tm}{\ovar}$ can be simulated by $\tm$.

\begin{theorem}[$\towhl$ strongly simulates $\topid$ via $\lamtopia{\cdot}{\ovar}$]
\label{l:reflection}
\hfill
\begin{enumerate}
\item \label{p:reflection-one} If $\lamtopia{\tm}{\ovar}\topidm\prtwo$ then exists $\tmtwo$ s.t. $\tm\towhldb\tmtwo$ and $\lamtopia{\tmtwo}{\ovar}\streq \prtwo$.
\item \label{p:reflection-two}If $\lamtopia{\tm}{\ovar}\topidls\prtwo$ then exists $\tmtwo$ s.t. $\tm\towhlls\tmtwo$ and $\lamtopia{\tmtwo}{\ovar}\streq \prtwo$.
\end{enumerate}
\end{theorem}

\proof
Both points are by induction on $\tm$. Cases:
\begin{itemize}
 \item \emph{Values}: if $\tm=\var$ or $\tm=\l \var.\tmthree$ then $\lamtopia{\tm}{\ovar}$ cannot reduce.
 \item \emph{Application}: if $\tm=\tmthree\tmfour$ then $\lamtopia{\tm}{\ovar}= \nu \ovartwo\nu \var  (\lamtopia{\tmthree}{\ovartwo}\paral \out{\ovartwo}{\var,\ovar}\paral!\inp{\var}{\ovarthree}.\lamtopia{\tmfour}{\ovarthree})$ with $\var$ fresh. Then:
 \begin{enumerate}
  \item \emph{Multiplicative reduction}. Cases of $\lamtopia{\tm}{\ovar}\topidm\prtwo$:
  \begin{itemize}
  \item \emph{Root}: $\lamtopia{\tmthree}{\ovartwo}=\nbvctxp{\varset\disunion\varsettwo}{\inp{\ovartwo}{\vartwo,\ovarfour}.\pr}$ with $\ovartwo\notin (\varset\disunion\varsettwo)$ and the process reduction is a $\topidm$ interaction with $\out{\ovartwo}{\var,\ovar}$ on $\ovartwo$. By Lemma \ref{l:conv-ctx}.\ref{p:conv-ctx-one} we get that $\Gamma=\emptyset$, $\tmthree=\sctxvp{\varset}{\l \vartwo.\tmthree'}$, and $\pr=\lamtopia{\tmthree'}{\ovarfour}$. So $\tm=\sctxvp{\varset}{\l \vartwo.\tmthree'}\tmfour$ and thus it has a $\towhldb$-redex on $\vartwo$, which maps to the $\topidm$ communication on $\ovartwo$ exactly as in the proof of Theorem \ref{tm:simulation}.\ref{p:simulation-one}.
  
  \item \emph{Inductive}: because of $\lamtopia{\tmthree}{\ovartwo}\topidm\prthree$. Then by \ih\ exists $\tmthree'$ s.t. $\tmthree\todb \tmthree'$ and $\lamtopia{\tmthree'}{\ovartwo}\streq \prthree$. We conclude by taking $\tmtwo:=\tmthree'\tmfour$.
  
  \end{itemize}
  \item  \emph{Exponential reduction}. $\lamtopia{\tm}{\ovar}\topidls\prtwo$ can only happen if reduction takes place in $\lamtopia{\tmthree}{\ovartwo}$, because $x$ is fresh by hypothesis. In such a case we conclude using the \ih, as in the first sub-case of the previous point.
 \end{enumerate}
 
\item \emph{Substitution}: if $\tm=\tmthree[\var/\tmfour]$ then $\lamtopia{\tm}{\ovar}= \nu \var (\lamtopia{\tmthree}{\ovar}\paral !\inp{\var}{\ovartwo}.\lamtopia{\tmfour}{\ovartwo})$. We have:
\begin{enumerate}
 \item \emph{Multiplicative reduction}. $\lamtopia{\tm}{\ovar}\topidm\prtwo$ can only happen if reduction takes place in $\lamtopia{\tmthree}{\ovar}$, and we conclude using the \ih.
 \item \emph{Exponential reduction}. If $\lamtopia{\tm}{\ovar}\topidls\prtwo$ because reduction takes place in $\lamtopia{\tmthree}{\ovar}$ we use the \ih. Otherwise, $\lamtopia{\tmthree}{\ovar}=\nbvctxp{\varset\disunion\varsettwo}{\out{\var}{\ovarthree}}$ with $\var\notin\varset\disunion\varsettwo$ and the process reduction is a $\topidls$ interaction with $!\inp{\var}{\ovartwo}.\lamtopia{\tmfour}{\ovartwo}$ on $\var$.
By Lemma \ref{l:conv-ctx}.\ref{p:conv-ctx-two} there exist $\varsetthree$ and $\evvctx{\varsetthree}$ s.t. $\tmthree=\evvctxp{\varsetthree}{\var}$. So $\tm=\evvctxp{\varsetthree}{\var}[\var/\tmfour]$ has a $\towhlls$ redex on $\var$, which maps to the $\topidls$ communication on $\var$ exactly as in the proof of Theorem \ref{tm:simulation}.\ref{p:simulation-two}.
\qed
\end{enumerate}
\end{itemize}

According to the two theorems of this section, the relationship between the call-by-name strategy on the ordinary $\l$-calculus and the evaluation in the $\pi$-calculus is the same as the relationship between the call-by-name strategy and linear weak head reduction. In the strong case (\ie\ when (head) reduction can go under lambdas), it is known that the latter can be at most quadratically longer than the former \cite{DBLP:conf/rta/AccattoliL12}. The analysis in \cite{DBLP:conf/rta/AccattoliL12} does not depend on being weak or strong. It follows that the same upper bound holds between the call-by-name strategy and its evaluation in the $\pi$-calculus.

Last, it is easy to see that linear weak head reduction is \emph{deterministic}: every term has at most one $\tohl$ redex, since every redex writes as $\evvctxp{\varset}{\val}$ (where $\val$ is a value, \ie\ a variable or an abstraction) and such a decomposition is unique. This property accounts for what Milner calls \emph{determinacy} of $\lamtopia{\tm}{\ovar}$ in \cite{DBLP:journals/mscs/Milner92}.

%% file: cbv.tex
\section{The call-by-value encoding}
\label{s:cbv}
We now show that the same exact relationship can be obtained with respect to call-by-value (CBV). The CBV calculus in use here is not Plotkin's calculus $\lv$. In \cite{DBLP:conf/flops/AccattoliP12} the author and Paolini introduced the \emph{value substitution calculus} $\lvsub$, which is a CBV calculus with explicit substitutions containing $\lv$ as a sub-calculus and behaving better than $\lv$ with respect to the semantical notion of \emph{solvability}. In \cite{AccLinearity,AccLSFA} we showed that $\lvsub$ has a sub-calculus, the \emph{value substitution kernel} $\lvker$, which has two key properties: 
\begin{enumerate}
\item \emph{Observational equivalence} \cite{AccLinearity}: there is a translation $\cdot^\circ:\lvsub\to\lvker$ s.t. $\tm$ and $\tm^\circ$ are equivalent with respect to observing any termination property.
\item \emph{Language for proof nets} \cite{AccLSFA}: $\lvker$ is an algebraic reformulation of the proof nets corresponding to the CBV translation of $\l$-calculus into \linlogic. Namely, there is a translation $\underline{\cdot}:\lvker\to PN$ which is a strong bisimulation.
\end{enumerate}
Here, we are going to show a further property: there are a CBV analogous $\tolwa$ of linear  weak head reduction $\towhl$ and a translation $\lamtopiv{\cdot}{\var}$ from $\lvker$ to the $\pi$-calculus which is a strong bisimulation with respect to $\tolwa$ and $\topid$. Let us  point out that in the untyped case there is a strong mismatch between Plotkin's calculus $\lv$ and the evaluation in \pns\ (see \cite{AccLinearity}), thus the results of this section do not hold with respect to $\lv$ (nor with any of its refinements with explicit substitutions where $\beta$-redexes are constrained to fire on values).

The \deff{value substitution kernel} $\lvker$ is given by the following grammar:
\begin{center}
$\begin{array}{rcl@{\hspace{3cm}}rcl}
\tm, \tmtwo, \tmthree,\tmfour & \grameq & \val \mid \val \tm\mid \tm[\var/\tmtwo] &
 \val &\grameq& \var \mid \l \var.\tm
\end{array}$ 
\end{center}

\noindent Please note that the left sub-term of an application can only be a value (see \cite{AccLinearity,AccLSFA} for more details). Substitution contexts $\sctxv{\varset}$ are defined as before. Instead, the language of \deff{evaluation contexts} changes: 
\begin{center}
$\begin{array}{rcl@{\hspace{3cm}}rcl}
\evvctx{\emptyset} &::=&  \ctxhole\mid \val\evvctx{\emptyset}\mid\tm[\var/\evvctx{\emptyset}]&
\evvctx{\varset\disunion\set{\var}} &::=& \evvctx{\varset}[\var/\tm] \mid \val\evvctx{\varset\disunion\set{\var}}\mid\tm[\vartwo/\evvctx{\varset\disunion\set{\var}}]
\end{array}$ 
\end{center}
Next, we define \deff{applicative contexts} as $\apctxvp{\varset}{\cdot}::=\evvctxp{\varset}{\ctxhole\tm}$.
As for CBN, we do not define the full calculus, but only the evaluation strategy. \deff{Linear weak applicative reduction}, noted $\tolwa$, is given by the rewriting rules $\tolwdb$ and $\tolwv$ defined as the closure by evaluation contexts of the following rules:
\begin{center}
$\begin{array}{lll@{\sep\sep\sep\sep\sep}lll@{\hspace{.6cm}}ll}
(\l \var.\tm) \tmtwo & \rtodb &  \tm[\var/\tmtwo]  &
 \apctxvp{\varset}{\var}[\var/\sctxvp{\varsetthree}{\val}] &\rtolwv & 
  \sctxvp{\varsetthree}{\apctxvp{\varset}{\val}[\var/\val]}& \mbox{$\var\notin\varset$}
\end{array}$ 
\end{center}

\noindent Note that the argument of a $\beta$-redex is not required to be a value, while the substitution rule can fire only in presence of a value (in a substitution context). As it was the case for the call-by-name calculus and for the $\pi$-calculus, one should also ask that $\fv{\val}\cap\varset=\emptyset$, $\fv{\apctxvp{\varset}{\var}}\cap\varsetthree=\emptyset$, and $\varset\cap\varsetthree=\emptyset$, but these side-conditions can always be satisfied by taking an $\alpha$-equivalent term, and so in the following they will be taken for granted. Note that $\var[\var/\vartwo] \not\rtolwv \vartwo$ but $(\var\varthree)[\var/\vartwo] \rtolwv \vartwo\varthree$, because substitution has to take place in an applicative context. This applicative restriction is a sort of converse to the head restriction used in the case of call-by-name evaluation. In terms of proof nets both these restrictions correspond to forbid reduction of cuts involving links in some $!$-boxes (with respect to the respective encodings of CBV 
and CBN), while the \emph{weak} requirement correspond to the analogous constraint with respect to the $\parr$-boxes mentioned in the introduction.
The applicative restriction is somehow a surprise, which is justified by the fact that it matches what happens in the $\pi$-calculus. It is a quite reasonable restriction: there is no point in substituting a value if it cannot be used in some application.

Linear weak applicative reduction enjoys a property which is the CBV analogous of the subterm porperty (deifned at the end of Section \ref{s:ls-calculus}). Let us call a \emph{v-subterm} a subterm which is a value.

\begin{lemma}[v-subterm property]
If $\tm \tolwa^k\tmtwo$ and $\val$ is a v-subterm of $\tmtwo$ then $\val$ is a v-subterm of $\tm$.
\end{lemma}

\begin{proof}
By induction on $k$. For $k=0$ it is trivial, for $k>0$ consider the term $\tmthree$ s.t. $\tmthree\tolwa\tmtwo$. The $\tolwdb$ rule does not create new values. The $\tolwv$ rule  duplicates a v-subterm of $\tmthree$, which by \ih\ is a v-subterm of $\tm$, and it does not substitute into v-subterms. So, any v-subterm of $\tmtwo$ is a v-subterm of $\tm$.
\end{proof}

Differently from linear weak head reduction, linear weak applicative reduction is a \emph{non-deterministic} stretegy: just consider $\tm=((\l \var.\var)(\vartwo \vartwo))[\vartwo/\varthree]$, which has two redexes. However, a simple induction shows that reduction is confluent: there is no need to use parallel reductions or other sophisticated techniques because no redex can duplicate/erase other redexes. In fact, it is easily seen that linear weak applicative reduction enjoys the diamond property. This fact corresponds to what Milner calls \emph{determinacy} of the CBV encoding.

\boldemph{The translation.} Similarly to the CBV translation of the $\l$-calculus to \linlogic, the CBV translation to the $\pi$-calculus uses an auxiliary function. The main translation function $\lamtopiv{\tm}{\var}$ is parametrized by a variable name $\var\notin\fv{\tm}$ (and not by a \specialv\ name) and the auxiliary function is noted $\lamtopiv{\cdot}{\ovar}$, \ie\ we use the same symbol but now the parameter is a \specialv\ name $\ovar$:
\begin{center}
$\begin{array}{rcl@{\sep\sep\sep\sep}rcl@{\sep}ll}
\lamtopiv{\val}{\var} &::=&  !\inp{\var}{\ovar}. \lamtopiv{\val}{\ovar}
&
\lamtopiv{\val\tmtwo}{\var} &::=& \nu \ovartwo\nu \vartwo  (\lamtopiv{\val}{\ovartwo}\paral \out{\ovartwo}{\vartwo,\var}\paral\lamtopiv{\tmtwo}{\vartwo}) & \mbox{$\vartwo$ is fresh}\\

\lamtopiv{\vartwo}{\ovar} &::=&  \out{\vartwo}{\ovar}
&
\lamtopiv{\tmtwo[\vartwo/\tmthree]}{\var} &::=&\nu \vartwo (\lamtopiv{\tmtwo}{\var}\paral \lamtopiv{\tmthree}{\vartwo})\\

\lamtopiv{\l \vartwo.\tmtwo}{\ovar} &::=&  \inp{\ovar}{\vartwo,\varthree}.\lamtopiv{\tmtwo}{\varthree}
\end{array}$ 
\end{center}

\noindent Note that the application case uses the auxiliary function on $\val$. Note also the difference with the call-by-name case: applications and explicit substitutions do not use replication, which is instead associated to values, with the important exception of applied values. The \emph{applicative} restriction on the strategy $\tolwa$ comes from this exception: the impossibility of interacting under replication in the $\pi$-calculus reflects on terms as the fact that one can substitute only on variables in applicative contexts, because the others are under a replication prefix. Last, this encoding is a minor variation over the CBV one in \cite{DBLP:conf/fossacs/ToninhoCP12}, which is not Milner's original CBV encoding.

\begin{lemma}
\label{l:cbv-transl-free-names}
Let $\tm\in\lvker$. Then $\fn{\lamtopiv{\tm}{\var}}=\fv{\tm}\disunion\set{\var}$ and $\fn{\lamtopiv{\tm}{\ovar}}=\fv{\tm}\disunion\set{\ovar}$.
\end{lemma}

\begin{proof}
By mutual induction on $\lamtopiv{\tm}{\var}$ and $\lamtopiv{\tm}{\ovar}$.
\end{proof}

The following lemma is the call-by-value analogous of Lemma \ref{l:wh-to-nb}.

\begin{lemma}[Relating $\evvctx{}$ and $\nbvctx{}$ via $\lamtopiv{\cdot}{\var}$]
\label{l:cbv-wh-to-nb}
Let $\varset$ be a set of variable names, $\var$ a variable name and $\evvctx{\varset}$ an evaluation context. There exist a set of names $\varsettwo$ (possibly containing both variables and \specialv\ names), a non-blocking context $\nbvctx{\varset\disunion\varsettwo}$, and a variable name $\varthree$ s.t. $\lamtopiv{\evvctxp{\varset}{\tm}}{\var}= \nbvctxp{\varset\disunion\varsettwo}{\lamtopiv{\tm}{\varthree}}$ and $\varsettwo\cap\fv{\tm}=\emptyset$ for every term $\tm$. Moreover, if $\evvctx{\varset}$ is a substitution context $\sctxv{\varset}$ then $\var=\varthree$, $\varsettwo=\emptyset$, and $\nbvctx{\varset}$ does not depend on $\var$.
\end{lemma}

\proof
By induction on $\evvctx{\varset}$. The base case is given by the empty context $\evvctx{\emptyset}=\ctxhole$, and it is trivial, just take $\varsettwo:=\emptyset$, $\nbvctx{\emptyset}:=\ctxhole$, and $\varthree:=\var$. The inductive cases:
\begin{itemize}
\item \emph{Right of an application}, $\evvctx{\varset}=\val\evvctxtwo{\varset}$:
\begin{center}
$\begin{array}{llllllllll}
\lamtopiv{\evvctxp{\varset}{\tm}}{\var} &=&\lamtopiv{\val \evvctxptwo{\varset}{\tm}}{\var} 
&=&\nu \ovartwo\nu \vartwo  (\lamtopiv{\val}{\ovartwo}\paral \out{\ovartwo}{\vartwo,\var}\paral\lamtopiv{\evvctxptwo{\varset}{\tm}}{\vartwo})\\

&&&=_{\ih}&\nu \ovartwo\nu \vartwo  (\lamtopiv{\val}{\ovartwo}\paral \out{\ovartwo}{\vartwo,\var}\paral\nbvctxptwo{\varset\disunion\varsetthree}{\lamtopiv{\tm}{\varthree}})&=& \nbvctxp{\varset\disunion\varsetthree\disunion\set{\vartwo,\ovartwo}}{\lamtopiv{\tm}{\varthree}}
\end{array}$
\end{center}
The \ih\ also gives $\varsetthree\cap\fv{\tm}=\emptyset$. Since $\ovartwo,\vartwo\notin\fv{\tm}$ it follows that $\varsettwo:=\varsetthree\disunion\set{\vartwo,\ovartwo}$ satisfies $\varsettwo\cap\fv{\tm}=\emptyset$.

\item \emph{Right of a substitution}, $\evvctx{\varset}=\tmtwo[\vartwo/\evvctxtwo{\varset}]$:
\begin{center}
$\begin{array}{llllllllll}
\lamtopiv{\evvctxp{\varset}{\tm}}{\var} &=&\lamtopiv{\tmtwo [\vartwo/\evvctxptwo{\varset}{\tm}]}{\var} 
&=&\nu \vartwo ( \lamtopiv{\tmtwo}{\var}\paral\lamtopiv{\evvctxptwo{\varset}{\tm}}{\vartwo})\\
&&&=_{\ih}&\nu \vartwo (\lamtopiv{\tmtwo}{\var}\paral\nbvctxptwo{\varset\disunion\varsetthree}{\lamtopiv{\tm}{\varthree}})
&=& \nbvctxp{\varset\disunion\varsetthree\disunion\set{\vartwo}}{\lamtopiv{\tm}{\varthree}}
\end{array}$
\end{center}

\noindent The \ih\ also gives $\varsetthree\cap\fv{\tm}=\emptyset$. Since $\vartwo\notin\fv{\tm}$ it follows that $\varsettwo:=\varsetthree\disunion\set{\vartwo}$ satisfies $\varsettwo\cap\fv{\tm}=\emptyset$.

\item \emph{Left of a substitution}, $\evvctx{\varset\disunion\set{\varthree}}=\evvctxtwo{\varset} [\vartwo/\tmthree]$. Then:
\begin{center}
$\begin{array}{llllllllll}
\lamtopiv{\evvctxp{\varset\disunion\set{\vartwo}}{\tm}}{\var} &=&\lamtopiv{\evvctxptwo{\varset}{\tm} [\vartwo/\tmthree]}{\var} 
&=&\nu \vartwo (\lamtopiv{\evvctxptwo{\varset}{\tm}}{\var}\paral \lamtopiv{\tmthree}{\vartwo})\\
&&&=_{\ih}&\nu \vartwo (\nbvctxptwo{\varset\disunion\varsettwo}{\lamtopiv{\tm}{\varthree}}\paral\lamtopiv{\tmthree}{\vartwo})
&=& \nbvctxp{\varset\disunion\set{\vartwo}\disunion\varsettwo}{\lamtopiv{\tm}{\varthree}}
\end{array}$
\end{center}

\noindent The \ih\ also gives $\varsettwo\cap\fv{\tm}=\emptyset$. Now, suppose that $\evvctx{\varset\disunion\set{\vartwo}}$ (and thus $\evvctxtwo{\varset}$) is a substitution context $\sctxv{\varset}$. Then by \ih\ we get $\nbvctxtwo{\varset}$ not depending on $\var$ s.t.:
\begin{center}
$\begin{array}{llllllllll}
\lamtopiv{\evvctxp{\varset\disunion\set{\vartwo}}{\tm}}{\var} &=&\lamtopiv{\evvctxptwo{\varset}{\tm} [\vartwo/\tmthree]}{\var} 
&=&\nu \vartwo (\lamtopiv{\evvctxptwo{\varset}{\tm}}{\var}\paral \lamtopiv{\tmthree}{\vartwo})\\
&&&=_{\ih}&\nu \vartwo (\nbvctxptwo{\varset}{\lamtopiv{\tm}{\var}}\paral\lamtopiv{\tmthree}{\vartwo})
&=& \nbvctxp{\varset\disunion\set{\vartwo}}{\lamtopiv{\tm}{\var}}
\end{array}$
\end{center}

\noindent where clearly $\nbvctx{\varset\disunion\set{\vartwo}}$ does not depend on $\var$.
\qed
\end{itemize}

\begin{theorem}[$\topi$ strongly simulates $\tolwa$]
\label{tm:cbv-simulation}
\hfill
\begin{enumerate}
\item \label{p:cbv-simulation-one}$\tm\tolwdb \tmtwo$ implies $\lamtopiv{\tm}{\var} \topidm \streq\lamtopiv{\tmtwo}{\var} $.
\item \label{p:cbv-simulation-two} $\tm\tolwv \tmtwo$ implies $\lamtopiv{\tm}{\var} \topidls\streq\lamtopiv{\tmtwo}{\var}$.
\end{enumerate}
\end{theorem}

\proof
We show the base cases, the inductive ones are as in the call-by-name case, using Lemma \ref{l:cbv-wh-to-nb}.
\begin{enumerate}
\item If $(\l \vartwo. \tm) \tmtwo\tolwdb \tm[\vartwo/\tmtwo]$ then:
\begin{center}
$\begin{array}{lllllllll}
\lamtopiv{(\l \vartwo. \tm) \tmtwo}{\var} &=&  \nu \ovartwo\nu \varthree  (\lamtopiv{\l \vartwo. \tm}{\ovartwo}\paral \out{\ovartwo}{\varthree,\var}\paral\lamtopiv{\tmtwo}{\varthree}) 

&=& \nu \ovartwo\nu \vartwo  (\inp{\ovartwo}{\vartwo,\varfour}.\lamtopiv{\tm}{\varfour}\paral \out{\ovartwo}{\varthree,\var}\paral\lamtopiv{\tmtwo}{\varthree}) \\

&\topidm& \nu \ovartwo \nu \vartwo  (\lamtopiv{\tm}{\varfour}\isub{\varfour}{\var}\isub{\vartwo}{\varthree}\paral \lamtopiv{\tmtwo}{\varthree}) 

&\alpeq&\nu \ovartwo \nu \vartwo  (\lamtopiv{\tm}{\var}\paral \lamtopiv{\tmtwo}{\vartwo}) \\
&=&\nu \ovartwo \lamtopiv{\tm[\var/\tmtwo]}{\var}
&\streq_{Lem. \ref{l:cbv-transl-free-names}}&\lamtopiv{\tm[\var/\tmtwo]}{\var}

\end{array}$
\end{center}

\item If  $\apctxvp{\varset}{\vartwo}[\vartwo/\sctxvp{\varsetthree}{\val}] \rtolwv  
  \sctxvp{\varsetthree}{\apctxvp{\varset}{\val}[\vartwo/\val]}$ and $\apctxvp{\varset}{\cdot}= \evvctxp{\varset}{\ctxhole\tmtwo}$ then:  
  \begin{center}
$\begin{array}{lll}
\lamtopiv{\apctxvp{\varset}{\vartwo}[\vartwo/\sctxvp{\varsetthree}{\val}]}{\var}

&=& \nu\vartwo(\lamtopiv{\evvctxp{\varset}{\vartwo\tmtwo}}{\var}\paral \lamtopiv{\sctxvp{\varsetthree}{\val}}{\vartwo})\\

&=_{Lem. \ref{l:cbv-wh-to-nb}}& \nu\vartwo(\nbvctxp{\varset\disunion\varsettwo}{\lamtopiv{\vartwo\tmtwo}{\varthree}}\paral \nbvctxptwo{\varsetthree}{\lamtopiv{\val}{\vartwo}})\\

&=& \nu\vartwo(\nbvctxp{\varset\disunion\varsettwo}{\lamtopiv{\vartwo\tmtwo}{\varthree}}\paral \nbvctxptwo{\varsetthree}{!\inp{\vartwo}{\ovar}.\lamtopiv{\val}{\ovar}})\\
&=&
\nu\vartwo(\nbvctxp{\varset\disunion\varsettwo}{\nu \ovartwo\nu \varfour  (\lamtopiv{\vartwo}{\ovartwo}\paral \out{\ovartwo}{\varfour,\varthree}\paral\lamtopiv{\tmtwo}{\varfour})
}\paral \nbvctxptwo{\varsetthree}{!\inp{\vartwo}{\ovar}.\lamtopiv{\val}{\ovar}})\\
&=&
\nu\vartwo(\nbvctxp{\varset\disunion\varsettwo}{\nu \ovartwo\nu \varfour  (\out{\vartwo}{\ovartwo}\paral \out{\ovartwo}{\varfour,\varthree}\paral\lamtopiv{\tmtwo}{\varfour})
}\paral \nbvctxptwo{\varsetthree}{!\inp{\vartwo}{\ovar}.\lamtopiv{\val}{\ovar}})\\
&\topidls&
\nu\vartwo\nbvctxptwo{\varsetthree}{\nbvctxp{\varset\disunion\varsettwo}{\nu \ovartwo\nu \varfour  (\lamtopiv{\val}{\ovartwo}\paral !\inp{\vartwo}{\ovar}.\lamtopiv{\val}{\ovar}\paral \out{\ovartwo}{\varfour,\varthree}\paral\lamtopiv{\tmtwo}{\varfour})}
}\\
&\streq_{Lem. \ref{l:pi-der-rules}.\ref{p:pi-der-rules-one}}&
\nu\vartwo\nbvctxptwo{\varsetthree}{\nbvctxp{\varset\disunion\varsettwo}{\nu \ovartwo\nu \varfour  (\lamtopiv{\val}{\ovartwo}\paral \out{\ovartwo}{\varfour,\varthree}\paral\lamtopiv{\tmtwo}{\varfour})}\paral !\inp{\vartwo}{\ovar}.\lamtopiv{\val}{\ovar}
}\\
&=&
\nu\vartwo\nbvctxptwo{\varsetthree}{\nbvctxp{\varset\disunion\varsettwo}{\lamtopiv{\val\tmtwo}{\varthree}}\paral !\inp{\vartwo}{\ovar}.\lamtopiv{\val}{\ovar}}\\
&=&
\nu\vartwo\nbvctxptwo{\varsetthree}{\lamtopiv{\evvctxp{\varset}{\val\tmtwo}}{\var}\paral !\inp{\vartwo}{\ovar}.\lamtopiv{\val}{\ovar}}\\
&\streq_{Lem. \ref{l:pi-der-rules}.\ref{p:pi-der-rules-four}}&
\nbvctxptwo{\varsetthree}{\nu\vartwo(\lamtopiv{\evvctxp{\varset}{\val\tmtwo}}{\var}\paral !\inp{\vartwo}{\ovar}.\lamtopiv{\val}{\ovar})}\\
&=&
\nbvctxptwo{\varsetthree}{\lamtopiv{\evvctxp{\varset}{\val\tmtwo}[\vartwo/\val]}{\var}}\\
&=_{Lem. \ref{l:cbv-wh-to-nb}}&
\lamtopiv{\sctxvp{\varsetthree}{\evvctxp{\varset}{\val\tmtwo}[\vartwo/\val]}}{\var}\\
&=&
\lamtopiv{\sctxvp{\varsetthree}{\apctxvp{\varset}{\val}[\vartwo/\val]}}{\var}
	\end{array}$
\end{center}

\noindent The $\streq$ step after the reduction is justified by the fact that $\ovartwo$, $\varfour$, and all the variables in $\varsettwo$ are introduced fresh and so do not belong to $\fv{\val}$. Moreover, $\varset\cap\fv{\val}=\emptyset$ by hypothesis,and so we can apply Lemma \ref{l:pi-der-rules}.\ref{p:pi-der-rules-one}.
\qed
\end{enumerate}

\paragraph{The converse relation.} As for call-by-name, we show that linear weak applicative reduction reflects exactly evaluation in the $\pi$-calculus.

\begin{lemma}
 \label{l:cbv-conv-ctx}
Let $\varset$ and $\varsettwo$ be a set of variable names and a set of \specialv\ names, respectively. Then:
\begin{enumerate}
\item \label{p:cbv-conv-ctx-one} If $\lamtopiv{\tm}{\var}=\nbvctxp{\varset\disunion\varsettwo}{!\inp{\var}{\ovar}.\pr}$ with $\var \notin\varset$ then $\varsettwo=\emptyset$ and exist $\val$ and $\sctxv{\varset}$ s.t. $\pr=\lamtopiv{\val}{\ovar}$ and $\tm=\sctxvp{\varset}{\val}$.

\item \label{p:cbv-conv-ctx-two} If $\lamtopiv{\tm}{\var}=\nbvctxp{\varset\disunion\varsettwo}{\out{\vartwo}{\ovar}}$ with $\vartwo\notin\varset$ then exist $\varsetthree\subseteq\varset$ and $\apctxv{\varsetthree}$ s.t. $\tm=\apctxvp{\varsetthree}{\vartwo}$.

\end{enumerate}
\end{lemma}

\proof
 Both points are by induction on $\tm$:
\begin{itemize}
\item \emph{Value}: if $\tm=\val'$ then $\lamtopiv{\tm}{\var}=!\inp{\var}{\ovar}.\lamtopiv{\val'}{\ovar}$.
\begin{enumerate}
\item Clearly $\Gamma=\Delta=\emptyset$, $\val$ is $\val'$, and $\sctxv{\varset}$ is the empty context.
\item The hypothesis is false, and so there is nothing to prove.
\end{enumerate}

\item \emph{Application}: if $\tm=\val'\tmtwo$ then $\lamtopiv{\val'\tmtwo}{\var}=\nu \ovartwo\nu \varthree  (\lamtopiv{\val'}{\ovartwo}\paral \out{\ovartwo}{\varthree,\var}\paral\lamtopiv{\tmtwo}{\varthree})$ with $\varthree$ and $\ovartwo$ are fresh.
\begin{enumerate}
\item By definition of the translation $\var\notin\fv{\val'\tmtwo}$ and so by Lemma \ref{l:cbv-transl-free-names} $\var\notin\fn{\lamtopiv{\val'}{\ovartwo}}\cup\fn{\lamtopiv{\tmtwo}{\varthree}}$. Consequently, there is no context $\nbvctx{\varset\disunion\varsettwo}$ s. t. $\lamtopiv{\tm}{\var}=\nbvctxp{\varset\disunion\varsettwo}{!\inp{\var}{\ovar}.\pr}$, so the hypothesis is false and there is nothing to prove.

\item Two cases:
\begin{enumerate}
\item $\lamtopiv{\val'}{\ovartwo}=\out{\vartwo}{\ovar}$ and $\nbvctx{\varset\disunion\varsettwo}=\nu \ovartwo\nu \varthree  (\ctxhole\paral \out{\ovartwo}{\varthree,\var}\paral\lamtopiv{\tmtwo}{\varthree})$, which imply $\val'=\vartwo$, $\ovar=\ovartwo$, $\varset=\set{\varthree}$, and $\varsettwo=\set{\ovartwo}$. We conclude taking $\varsetthree:=\emptyset$ and $\apctxv{\emptyset}:=\ctxhole\tmtwo$.

\item The context hole $\ctxhole$ is in $\lamtopiv{\tmtwo}{\varthree}$. Let $\varset':=\varset\setminus\set{\varthree}$ and $\varsettwo':=\varsettwo\setminus\set{\ovartwo}$. If $\lamtopiv{\tm}{\var}=\nbvctxp{\varset\disunion\varsettwo}{\out{\varthree}{\ovar}}$ then $\lamtopiv{\tmtwo}{\var}=\nbvctxptwo{\varset'\disunion\varsettwo'}{\out{\varthree}{\ovar}}$ for some context $\nbvctxtwo{\varset'\disunion\varsettwo'}$. The \ih\ gives $\varsetthree\subseteq\varset'$ and an applicative context $\apctxvtwo{\varsetthree}$ s.t. $\tmtwo=\apctxvptwo{\varsetthree}{\vartwo}$. We conclude taking  $\apctxv{\varsetthree}:=\val'\apctxvtwo{\varsetthree}$. 
\end{enumerate}
\end{enumerate}

\item \emph{Substitution}: if $\tm=\tmtwo[\varthree/\tmthree]$ then $\lamtopiv{\tm}{\var}= \nu \varthree (\lamtopiv{\tmtwo}{\var}\paral \lamtopiv{\tmthree}{\varthree})$.
\begin{enumerate}
\item  By definition of the translation $\var\notin\fv{\tmtwo[\varthree/\tmthree]}$ and so by Lemma \ref{l:cbv-transl-free-names} $\var\in\fn{\lamtopiv{\tmtwo}{\var}}$ and $\var\notin\fn{\lamtopiv{\tmthree}{\varthree}}$. Consequently, the context hole $\ctxhole$ is in $\lamtopiv{\tmtwo}{\var}$, which then writes as $\nbvctxptwo{\varset'\disunion\varsettwo}{!\inp{\var}{\ovar}.\pr}$, with $\varset=\varset'\disunion\set{\varthree}$ for some context $\nbvctxtwo{\varset'\disunion\varsettwo}$. By \ih\ we get that there exist $\val$ and $\sctxvtwo{\varset'}$ s.t. $\pr=\lamtopiv{\val}{\ovar}$ and $\tmtwo=\sctxvptwo{\varset'}{\val}$. We conclude taking $\sctxv{\varset}:=\sctxvtwo{\varset'}[\varthree/\tmthree]$.

\item Two cases:
\begin{enumerate}
\item The context hole $\ctxhole$ is in $\lamtopiv{\tmtwo}{\var}$. Let $\varset':=\varset\setminus\set{\varthree}$. If $\lamtopiv{\tm}{\var}=\nbvctxp{\varset\disunion\varsettwo}{\out{\varthree}{\ovar}}$ then $\lamtopiv{\tmtwo}{\var}=\nbvctxptwo{\varset'\disunion\varsettwo}{\out{\varthree}{\ovar}}$ for some context $\nbvctxtwo{\varset'\disunion\varsettwo}$. The \ih\ gives $\varsetthree'\subseteq\varset'$ and an applicative context $\apctxvtwo{\varsetthree'}$ s.t. $\tmtwo=\apctxvptwo{\varsetthree'}{\vartwo}$. We conclude taking $\varsetthree:=\varsetthree'\disunion\set{\varthree}$ and $\apctxv{\varsetthree}:=\apctxvtwo{\varsetthree'}[\varthree/\tmthree]$. 
\item The context hole is in $\lamtopiv{\tmthree}{\varthree}$. Analogous to the previous case (except that $\varsetthree=\varsetthree'$).
\qed
\end{enumerate}
\end{enumerate}
\end{itemize}

\begin{theorem}[$\tolwa$ strongly simulates $\topid$ via $\lamtopiv{\cdot}{\ovar}$]
\label{l:cbv-reflection}
\hfill
\begin{enumerate}
\item \label{p:cbv-reflection-one} If $\lamtopiv{\tm}{\var}\topidm\prtwo$ then exists $\tmfour$ s.t. $\tm\tolwdb\tmfour$ and $\lamtopiv{\tmfour}{\var}\streq \prtwo$.
\item \label{p:cbv-reflection-two}If $\lamtopiv{\tm}{\var}\topidls\prtwo$ then exists $\tmfour$ s.t. $\tm\tolwv\tmfour$ and $\lamtopiv{\tmfour}{\var}\streq \prtwo$.
\end{enumerate}
\end{theorem}

\proof
By induction on $\tm$. Cases:
\begin{itemize}
 \item \emph{Values}: if $\tm$ is a value then $\lamtopiv{\tm}{\var}$ cannot reduce.
 \item \emph{Application}: if $\tm=\val\tmtwo$ then $\lamtopiv{\val\tmtwo}{\var} = \nu \ovartwo\nu \vartwo  (\lamtopiv{\val}{\ovartwo}\paral \out{\ovartwo}{\vartwo,\var}\paral\lamtopiv{\tmtwo}{\vartwo})$ with $\vartwo$ and $\ovartwo$ fresh. Then:
 \begin{enumerate}
  \item \emph{Multiplicative reduction}. Cases of $\lamtopia{\tm}{\var}\topidm\prtwo$:
  \begin{itemize}
    \item \emph{Root}: $\lamtopiv{\val}{\ovartwo}=\inp{\ovartwo}{\varthree,\varfour}.\pr$ interacts with $\out{\ovartwo}{\vartwo,\var}$ on $\ovartwo$. Clearly, $\val$ is an abstraction $\l \varthree.\tmthree$ with $\lamtopiv{\tmthree}{\varfour}=\pr$, and $\tm=(\l \varthree.\tmthree) \tmtwo$ has a root $\tolwdb$ redex. Then, $\tm$ and $\lamtopiv{\tm}{\var}$ are related exactly as in the proof of Theorem \ref{tm:cbv-simulation}.\ref{p:cbv-simulation-one}. Note that $\ovartwo\notin\fn{\lamtopiv{\tmtwo}{\vartwo}}$ by Lemma \ref{l:cbv-transl-free-names}, and so there cannot be any multiplicative root interaction involving $\lamtopiv{\tmtwo}{\vartwo}$.
    \item \emph{Inductive}: $\lamtopiv{\tm}{\var}\topidm\prtwo$ because $\lamtopiv{\tmtwo}{\vartwo}\topidm\pr$. By \ih\ we get that there exists $\tmfour'$ s.t. $\tmtwo\todb\tmfour'$ and $\lamtopiv{\tmfour'}{\vartwo}\streq\pr$. Since $\val\ctxhole$ is an evaluation contexts, taking $\tmfour:=\val\tmfour'$ we get $\tm\todb\tmfour$ and $\lamtopiv{\tmfour}{\var}\streq\pr$.
  \end{itemize}
  \item \emph{Exponential reduction}. The inductive case (\ie\ $\lamtopiv{\tm}{\var}\topidls\prtwo$ because $\lamtopiv{\tmtwo}{\vartwo}\topidm\pr$) follows by the \ih\ as in the inductive case for multiplicative reductions. In the root case there cannot be any root exponential reduction. Indeed, $\lamtopiv{\val}{\ovartwo}$ would have to be $\out{\varthree}{\ovartwo}$ and $\lamtopiv{\tmtwo}{\vartwo}$ should have a $!\inp{\varthree}{\ovarthree}.\pr$ sub-process. This second requirement is only possible if $\tmtwo$ contains a value $\val$ which in $\lamtopiv{\tmtwo}{\vartwo}$ is translated with respect to $\varthree$, so that $\lamtopiv{\val}{\varthree}=!\inp{\varthree}{\ovarthree}.\pr$. But this is impossible because $\vartwo$ is fresh (and so $\vartwo\neq\varthree$) and any variable name which is used as a parameter in the translation of a subterm of $\tmtwo$ is either $\vartwo$ or it is introduced fresh (and so cannot be equal to $\varthree$).
 \end{enumerate}
\item \emph{Substitution}: if $\tm=\lamtopiv{\tmtwo[\vartwo/\tmthree]}{\var}$ then $\lamtopiv{\tm}{\var}=\nu \vartwo (\lamtopiv{\tmtwo}{\var}\paral \lamtopiv{\tmthree}{\vartwo})$

\begin{enumerate}
  \item \emph{Multiplicative reduction}. If the reduction takes place in $\lamtopiv{\tmtwo}{\var}$ or $\lamtopiv{\tmthree}{\vartwo}$ we use the \ih\ as in the previous inductive cases. And there cannot be any root multiplicative reduction. Indeed, it should be along a \specialv\ name $\ovar$ free in both $\lamtopiv{\tmtwo}{\var}$ and $\lamtopiv{\tmthree}{\vartwo}$, but by Lemma \ref{l:cbv-transl-free-names} $\lamtopiv{\tmtwo}{\var}$ and $\lamtopiv{\tmthree}{\vartwo}$ have no free \specialv\ name.

  \item \emph{Exponential reduction}. If the reduction takes place in $\lamtopiv{\tmtwo}{\var}$ or $\lamtopiv{\tmthree}{\vartwo}$ we use the \ih\ as in the previous inductive cases.
 
 Otherwise, an exponential reduction can only be along a variable name $\varthree$ which is free in both $\lamtopiv{\tmtwo}{\var}$ and $\lamtopiv{\tmthree}{\vartwo}$. Then $\varthree\neq\var$, because $\var\notin\fn{\lamtopiv{\tmthree}{\vartwo}}$. Another requirement is that $\varthree$ has to be used as the parameter of the translation of a value $\val$, which is the only way to get a replicated input. The only possibility then is that $\varthree=\vartwo$, because all variable parameter names used in the translation and different from $\var$ and $\vartwo$ are fresh and cannot be in both $\lamtopiv{\tmtwo}{\var}$ and $\lamtopiv{\tmthree}{\vartwo}$.
    
    Now, $\lamtopiv{\tmtwo}{\var}$ has to be of the form $\nbvctxp{\varset\disunion\varsettwo}{\out{\vartwo}{\ovar}}$ and $\lamtopiv{\tmthree}{\vartwo}$ has to be of the form $\nbvctxptwo{\varset'\disunion\varsettwo'}{!\inp{\vartwo}{\ovartwo}.\pr}$, for some sets of variable names $\varset$ and $\varset'$ and some sets of \specialv\ names $\varsettwo$ and $\varsettwo'$, and with $\vartwo\notin\varset\cup\varset'$. By Lemma 
    \ref{l:cbv-conv-ctx} we get $\varsettwo'=\emptyset$ and that exist $\val$, $\sctxv{\varset'}$, $\varsetthree\subset\varset$, and $\apctxv{\varsetthree}$ s.t. $\pr=\lamtopiv{\val}{\ovartwo}$, $\tmthree=\sctxvp{\varset'}{\val}$, and $\tmtwo=\apctxvp{\varsetthree}{\vartwo}$. Summing up, $\tm=\apctxvp{\varsetthree}{\vartwo}[\vartwo/\sctxvp{\varset'}{\val}]$ and it has a $\tolwv$ redex which maps on $\lamtopia{\tm}{\var}\topidls\prtwo$ exactly as in the proof of Theorem \ref{tm:cbv-simulation}.\ref{p:cbv-simulation-two}.
  \qed
  \end{enumerate}  
  \end{itemize}

%